\documentclass[journal,comsoc]{IEEEtran}

\usepackage{color}
\usepackage{xcolor}
\usepackage{amsmath}
\usepackage{multirow}
\usepackage{amssymb}
\usepackage{algorithm}
\usepackage[noend]{algpseudocode}

 \usepackage{amsthm}

\newtheorem{proposition}{Proposition}
\newtheorem{lemma}{Lemma}
\newtheorem{remark}{Remark}

\usepackage{subcaption}
 \usepackage{graphicx}

\usepackage{cite}
\ifCLASSINFOpdf
\else
\fi

\begin{document}
\title{Trajectory Design and Power Allocation for Drone-Assisted NR-V2X Network with Dynamic NOMA/OMA}
\author{Omid~Abbasi, \IEEEmembership{Student Member,~IEEE}, Halim~Yanikomeroglu, \IEEEmembership{Fellow,~IEEE}, Afshin~Ebrahimi, \IEEEmembership{Member,~IEEE}, and Nader~Mokari, \IEEEmembership{Senior Member,~IEEE} \thanks{This work was supported in part by Huawei Canada Co.,
Ltd. (\textit{Corresponding author: Afshin Ebrahimi})}
\thanks{O. Abbasi and A. Ebrahimi are with the Department of Electrical Engineering, Sahand University of Technology, Tabriz, Iran. O. Abbasi is also
with the Department of Systems and Computer Engineering, Carleton
University, Ottawa, ON K1S5B6, Canada. (e-mail: om\_abbasi@sut.ac.ir; omidabbasi@sce.carleton.ca; aebrahimi@sut.ac.ir)}\thanks{H. Yanikomeroglu is with the Department of Systems and Computer
Engineering, Carleton University, Ottawa, ON K1S5B6, Canada. (e-mail:
halim@sce.carleton.ca)}\thanks{N. Mokari is with the Department of Electrical and Computer Engineering,
Tarbiat Modares University, Tehran, Iran. (e-mail:
nader.mokari@modares.ac.ir)}}
\maketitle

\begin{abstract}
  In this paper, we find trajectory planning and power allocation for a vehicular network in which an unmanned-aerial-vehicle (UAV) is considered as a relay to extend coverage for two disconnected far vehicles. We show that in a two-user network with an amplify-and-forward (AF) relay, non-orthogonal-multiple-access (NOMA) always has better or equal sum-rate in comparison to orthogonal-multiple-access (OMA) at high signal-to-noise-ratio (SNR) regime. However, for the cases where i) base station (BS)-to-relay link is weak, or ii) two users have similar links, or iii) BS-to-relay link is similar to relay-to-weak user link, applying NOMA has negligible sum-rate gain. Hence, due to the complexity of successive-interference-cancellation (SIC) decoding in NOMA, we propose a dynamic NOMA/OMA scheme in which OMA mode is selected for transmission when applying NOMA has only negligible gain. Also, we show that OMA always has better min-rate than NOMA at high SNR regime. Further, we formulate two optimization problems which maximize the sum-rate and min-rate of the two vehicles. These problems are non-convex, and hence we propose an iterative algorithm based on alternating-optimization (AO) method which solves trajectory and power allocation sub-problems by successive-convex-approximation (SCA) and difference-of-convex (DC) methods, respectively. Finally, the above-mentioned performance is confirmed by simulations.

\end{abstract}

\begin{IEEEkeywords}
Non-orthogonal multiple access, aerial relaying, dynamic multiple access, trajectory design, power allocation, V2X.
\end{IEEEkeywords}

% \fi
\IEEEpeerreviewmaketitle

\section{Introduction}
\label{intro}
Unmanned aerial vehicle (UAV) communication has attracted substantial attention in the context of 5G and beyond-5G due to its many advantages such as high mobility, low cost, and on-demand deployment \cite{valavanis2015future}. UAV-enabled communications have some advantages over terrestrial wireless communications. Due to the higher possibility of the existence of the line-of-sight (LOS) link between UAV and ground users \cite{lin2018sky}, there is less path loss for these links. Also, multipath fading rarely occurs in UAV-to-ground links. UAVs can be used as aerial base stations (BS) \cite{irem} and as relays \cite{zeng2016throughput} in practical scenarios that require
on-demand deployment \cite{zhan2011wireless}.\par
Vehicular communication has gained much attention in recent years due to its potential to enhance road safety and traffic efficiency \cite{vehicular-survey}. Furthermore, providing communication for vehicles improves the quality of infotainment (information and entertainment) experience in vehicles \cite{cheng2011infotainment}. This technology also facilitates emerging autonomous drivers. These applications require a considerable amount of data exchange, high capacity, and low latency. IEEE 802.11p is a standard which adds wireless access for vehicular environment (WAVE) \cite{wave} as dedicated-short-rang-communication (DSRC). However, IEEE 802.11p networks have coverage problem for vehicle-to-infrastructure (V2I) connections \cite{d2d-v}. Cellular-vehicle-to-everything (C-V2X) has advantages of high data rates and greater coverage than WAVE \cite{cv2x-noma}. Recently, IEEE 802.11p and C-V2X technologies are updated to IEEE 802.11bd and new radio (NR)-V2X to satisfy high reliability and low latency requirements \cite{NR-V2X}.  \par
% NOMA uses power domain to separate users to improve spectrum efficiency and reduce latency. Indeed, by increasing the number of users served with the same resources, the probability of resource collision reduces and hence delay decreases. 
  \par

\subsection{Related Works}
Non-orthogonal multiple access (NOMA) \cite{saito2013non} has attracted significant attention during 5G standardization \cite{ding2016application}. Also, cooperative NOMA has been considered for terrestrial networks in recent years. In \cite{ding2015cooperative}, near users acted as relay for far users. In \cite{liu2016cooperative}, the authors proposed a new cooperative NOMA protocol in which near NOMA users act as energy harvesting relays to help far users.  The authors in \cite{abbasi2018cooperative} studied the performance of a cooperative NOMA system with a dedicated full-duplex relay which worked in amplify-and-forward (AF) mode. The impact of relay selection on the performance of a cooperative NOMA system was studied in \cite{ding2016relay}. In \cite{kim2015non}, the authors exploited a  decode-and-forward relay to send information to far users. The authors of \cite{kim2015non} derived closed-form expressions for outage probability and ergodic sum-rate of NOMA users. The authors in \cite{noma-inspired} proposed a cooperative relaying
system based on NOMA
which uses an AF relay.   \par
 In UAV-aided relay networks,
UAVs are deployed to extend coverage and provide wireless connectivity for distant users without reliable direct
communication links. In \cite{choi2014energy}, an algorithm for energy efficiency maximization was proposed in which aerial relay moves in a circular
trajectory. The authors in \cite{zeng2016throughput} considered a novel mobile relaying
technique, where the relay nodes are mounted on UAVs. They studied the throughput
maximization problem in mobile relaying systems by optimizing
the source/relay transmit power along with the relay trajectory.
In \cite{zhan2011wireless}, the authors studied a network that consists of a BS and several aerial relays which serve several ground
users. The deployment of UAVs for intercellular traffic offloading was studied in \cite{rohde2013ad}. In \cite{zhang2018joint}, the authors considered a UAV relay network, where the UAV works as an
AF relay. In \cite{xu-joint}, a UAV was employed as a relay to improve fairness and energy efficiency.  In \cite{access-af-aerial}, the authors considered maximizing the throughput of a mobile relay system employing power allocation and UAV trajectory planning. \par
There are a few studies that exploit NOMA for a UAV-mounted BS to
serve terrestrial users \cite{Nallanathan-noma-uav}.  
%In \cite{sharma2017uav,sohail2018non,poor-uav-noma,f.r.yu,alo-precoding},  the authors considered a NOMA-based system with a UAV-mounted BS to
%serve terrestrial users. 
In \cite{sharma2017uav}, the authors deployed a fixed-wing type UAV which moves in a circular
trajectory around the centre of a macro-cell to provide
coverage to the ground users with NOMA scheme. In \cite{sohail2018non}, the authors proposed a power allocation scheme to
maximize the sum-rate of the NOMA system
by reducing energy expense for the UAV. Authors in \cite{poor-uav-noma}, considered
a multi-user system, in which a single-antenna
UAV-BS serves a large number of ground users by employing NOMA.  In \cite{f.r.yu}, the placement and
power allocation were jointly optimized to improve the
performance of the NOMA-UAV network. In \cite{alo-precoding},
a UAV and BS cooperate to serve
ground users simultaneously. In this paper, the sum rate was maximized by
jointly optimizing the UAV trajectory and the NOMA precoding. \par
%The authors in \cite{Uplink-Cooperative-NOMA} applied NOMA to the uplink communication from a UAV to cellular BSs, under spectrum sharing with the  ground users. 
Vehicular communication has attracted much attention in recent years. In \cite{d2d-v}, the authors investigated the spectrum sharing and
power allocation design of device-to-device (D2D)-enabled vehicular networks. In \cite{cv2x-noma}, two NOMA-based relay-assisted broadcasting
and multicasting schemes were proposed for 
C-V2X communications.
In \cite{NOMA-Enabled-V2X}, the resource allocation problem for NOMA-enabled V2X communications
was investigated. The authors in \cite{cache-v2x} introduced cache-aided NOMA as an enabling technology for
vehicular networks. 

\subsection{Motivations and Contributions}
To the best of our knowledge, there are no works in the related literature that have considered aerial communication for vehicular networks. In this paper, we consider a vehicular network in which a ground BS serves disconnected terrestrial vehicles with the aid of an aerial relay. 
%Note that due to establishing LOS links by drones, the channel between the drones and vehicles varies slowly in contrast to the vehicle to terrestrial BS links in which channel varies fast due to the mobility of vehicles and multipath fading. Hence, the large overhead of channel estimation in vehicular communication which leads to high latency is reduced by exploiting drones.
%In our system, the BS sends the messages of vehicles based on a dynamic NOMA/OMA scheme.
Due to mitigating resource collision, NOMA reduces latency and therefore is a good choice for vehicular communication. Note that applying NOMA is straightforward for air-to-ground channels rather than for ground-to-ground ones. Indeed, multipath fading is not dominant in air-to-ground links and there is less randomness in these links.  Hence, in order to apply NOMA scheme, we just need to sort the vehicles based on their distances.\par
%In this paper, we show that in a two-user network with a relay and with high SNRs, NOMA is superior to OMA in the sum-rate performance at two cases:
  % 1) BS to relay link is stronger than relay to user links, and two users have degraded links. 
   %2) BS to relay link is better than relay to weak user link. Also, we show that OMA has always better min-rate performance in comparison to NOMA at the high SNR regime. Based on these observations and due to mobility of vehicles and UAV which changes channel power gains between nodes, we propose a dynamic NOMA/OMA scheme in which NOMA or OMA mode is selected for transmission according to the quality of communication links.
   In this paper, we show that in a two-user network with an AF relay, NOMA always has better or equal sum-rate performance in comparison to OMA at high SNR regime. However, for the cases where i) BS-to-relay link is weak, or ii) two users have similar links, or iii) BS-to-relay link is similar to relay-to-weak user link, applying NOMA has negligible sum-rate gain. Hence, due to the complexity of successive-interference-cancellation (SIC) decoding in NOMA, we propose a dynamic NOMA/OMA scheme in which OMA is selected for transmission when applying NOMA has negligible gain. 
   %Note that in the case where OMA is selected, the decoding complexity decreases without compromising the throughput.
   In the proposed scheme, both vehicles can apply SIC based on the quality of the channel between these vehicles and the relay node.  Also, we show that OMA always has better min-rate performance than NOMA at high SNR regime.\par
We consider two scenarios for this network. In the first scenario, due to high capacity requirements of V2I links, we formulate an optimization problem which maximizes the sum-rate of two vehicles at all time slots satisfying the required rates of each vehicle at each time slot. This scenario is suitable for delay-tolerant cases, where for example both vehicles are downloading a video. We find optimal values for UAV trajectory, transmit powers of NOMA vehicles at the BS, and transmit power of the relay node to maximize this sum-rate. In the second scenario, in order to provide more uniform and fair rate performance between two users and among all time slots, we optimize the minimum rate of users at each time slot. This scenario is for delay-constrained applications such as safety-critical services. These optimization problems are non-convex and intractable to solve. In order to solve the above-mentioned non-convex problems, we apply the alternating optimization (AO) method. Hence, we divide our optimization problem into two separate sub-problems. In the first sub-problem, with a given trajectory planning, we optimize the transmit power of vehicles and the relay. With some manipulations, this problem is a difference of concave (DC) programming problem. In the second sub-problem, the trajectory of the UAV is optimized for given power allocations. Both of these sub-problems are still non-convex, and hence we apply the successive convex approximation (SCA) method to solve them. \par
Our system model can be generalized into a multi-user system with $K$ users. We can serve these disconnected vehicles by two methods. At the first method, we assign one UAV for each pair of NOMA users. Hence, we require $K/2$ UAVs to serve $K$ users, and each of these UAVs must work in different sub-carriers. At the second method, we assume that we have only one UAV for all of the $K$ users. In order to have less decoding complexity, these $K$ users are divided into $K/2$ groups with two users inside each group, and then NOMA scheme is applied for each pair. These pairs are distinguished by different sub-carriers. Note that in order to generalize our formulated optimization problems and their solutions for multi-user case, one can define user paring coefficients for NOMA scheme to pair vehicles. 
%Proposed solutions for trajectory design and power allocations in our paper can be applied into these new problems. 
The problem of user pairing for NOMA scheme has been investigated well in the literature \cite{pairing} and is beyond the scope of this paper.\par
The main contributions of this paper are summarized as follows:
\begin{itemize}
  %We prove that in a typical two-user network with a dedicated relay and with high SNRs, NOMA is superior to OMA in sum-rate performance at two cases:
  % 1) BS to relay link is stronger than relay to user links, and two users have degraded links. 
  % 2) BS to relay link is better than relay to weak user link. Also, we show that OMA has always better min-rate performance in comparison to NOMA at the high SNR regime.
  \item We show that in a two-user network with a dedicated AF relay, NOMA always has better or equal sum-rate performance in comparison to OMA at high SNR regime.  Also, we show that OMA always has better min-rate performance than NOMA at high SNR regime. 
   \item We show that for the cases where i) BS-to-relay link is weak, or ii) two users have similar links, or iii) BS-to-relay link is similar to relay-to-weak user link, applying NOMA has negligible sum-rate gain in comparison to OMA. Hence, due to the complexity of SIC decoding at NOMA, we propose a dynamic NOMA/OMA scheme in which OMA is selected for transmission when applying NOMA has only negligible gain. 
  \item We formulate two optimization problems in which the sum-rate and min-rate of two vehicles are maximized. The formulated problems are non-convex. Hence, using the AO method we divide the original problem into two separate sub-problems for optimizing trajectory and power allocations. These sub-problems are still non-convex, and hence we solve them via the SCA  and DC methods. The proposed efficient method converges at few iterations.

\end{itemize}

\subsection{Organization}
The remainder of this paper is organized as follows. Section II presents the system model. Section III presents the proposed dynamic NOMA/OMA. Sections IV and V provide the formulated problems for the sum-rate and min-rate problems. Section VI provides simulation
results to validate the performance of the proposed algorithms. Finally, Section VII concludes the paper.





\section{SYSTEM MODEL}
The considered system model consists of a BS, a UAV-mounted relay, and two disconnected NOMA vehicles as depicted in Fig. \ref{fig:model}. The vehicles receive different information (e.g., vehicle-specific control
information) from the BS.  We assume that the locations of vehicles are known a priory for the BS. Note that because vehicles are typically moving in a straight line at a street, this assumption is reasonable. Also, for some future applications, like autonomous driving, a central unit like a BS is responsible to control and drive vehicles, and hence it is aware of locations of vehicles. Due to high computational capabilities of BS rather than other nodes in our system model, we assume that the BS performs optimization algorithms.
%Then, BS transmits allocated power of relaying to the UAV and multiple access mode to the UAV and vehicles.
%Note that in some applications, like autonomous driving, a central unit controls the speed and path of the car, and hence the assumption about the BS' knowledge of the car is rational.  
Due to existing obstacles, blocking buildings, and significant path loss, there is no direct link between the BS and the vehicles, and the UAV helps to extend coverage and establish a communication link for these disconnected vehicles. Note that we assume that vehicles can not establish communication link with other BSs, and therefore there is no inter-cell interference from other cells in our system model.  We assume that the relay node applies the AF protocol and works in half-duplex mode.  In AF protocol, the relay only amplifies and retransmits the received signal, and so the complexity of computations is reduced \cite{patel2006statistical}. Note that in \cite{zeng2016throughput} the relay node buffers messages until a good communication link is established. The scheme in \cite{zeng2016throughput} adds an extra buffering delay to the system which is not desirable for ultra reliable low latency communication (URLLC). However, in this paper, we assume that at each time the relay node retransmits the received signal which imposes less delay in the transmission of data. All of the nodes are assumed to be equipped with a single antenna. \par
\begin{figure}[t]
\centering
  \includegraphics[width=\linewidth]{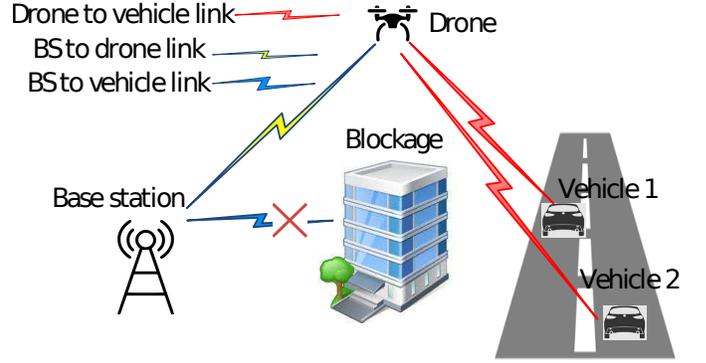}
  \caption{System model for drone-assisted NR-V2X network.}
  \label{fig:model}
\end{figure}
Hereafter, we show the BS, relay, vehicle 1, and vehicle 2 with subscripts s, r, 1, and 2, respectively. In our system, the BS is placed on the origin, and its coordinates is denoted by $(x_{BS},y_{BS},h_{BS})=(0, 0, 0)$. Also, the coordinates of relay node, vehicle 1, and vehicle 2 are denoted by $(x[n], y[n], h[n])$, $(x_{1}[n], y_{1}[n], 0)$ and $(x_{2}[n], y_{2}[n], 0)$, respectively, where $n$ indicates the number of time slots. Moreover, we assume that the UAV flies at a fixed height $h[n]=h$ \cite{zeng2016throughput}. We assume that there are predefined initial and final locations for the UAV flight path.  Hence, the constraints $(x[0], y[0])=(x_{s}, y_{s})$, and $(x[N+1], y[N+1])=(x_{f}, y_{f})$ must be applied in the optimization problems where $N$ is the total number of time slots.
We assume that the flight time $T$ for the UAV is divided into $N$ slots, such that $T=N\tau$, and $\tau$ is supposed to be small enough. Therefore, during each time slot, the position of the relay node is fixed \cite{zeng2016throughput}. 
% Note that $\tau$ can be a greater value in our proposed air-to-vehicle links rather than conventional ground-to-vehicle links. Indeed, due to existing LOS link between the drone and vehicles, the channel varies slower than conventional vehicular networks. This means that we need to update the estimation of the channel in greater time intervals which reduces the overhead of the channel estimation. Consequently, we can transmit messages with fewer bits and smaller packets sizes which causes low latency. 
In this paper, we assume that during flight time $T$, the UAV flies from position $(x_{s}, y_{s},h)$  to position $(x_{f}, y_{f},h)$. Also, the UAV can fly with the maximum speed denoted by $V$. Hence, during each time slot, the relay moves on the basis of the velocity constraint \cite{zeng2016throughput} as
$(x[n]-x[n-1])^{2}+(y[n]-y[n-1])^{2}\leq(V\tau)^{2},~~\forall n$, which shows that due to the existence of a maximum velocity limitation, maximum displacement of the UAV must be less than $V\tau$ at each time slot $n$.

The channel power gains of the BS to the relay, the relay to vehicle 1, and the relay to vehicle 2 are indicated by $h_{r}[n]$, $h_{1}[n]$, and $h_{2}[n]$, respectively.  We assume that there is no small scale fading for the air-to-ground channels. We suppose that these channels are dominated by LOS links, and hence multipath fading can be ignored \cite{zeng2016throughput}. Hence, the channel power $h_{i}[n]$ follows the free space path loss model as
\begin{equation}\label{channel-model}
h_{i}[n]=\beta_{0}d_{i}^{-2}[n]=\dfrac{\beta_{0}}{(x[n]-x_i[n])^{2}+(y[n]-y_{i}[n])^{2}+h^{2}},
\end{equation}
  for $n=1,...,N,~i=1,2$, where $\beta_{0}$ denotes the channel power at the reference distance $d_{0}=1~\mathrm{m}$. We consider independent additive white gaussian noise (AWGN) $z_{k}[n]$ with the distribution $CN(0,\sigma_{k}^{2}[n])$ ($k \in\{r, 1, 2\},\forall n$) in which $\sigma_{k}^{2}[n]=\sigma^{2}$ shows the variance of the noise for node $k$ at time slot $n$ \footnote{Note that for thermal noise at receiver we have $N_0=KT$, in which $K$ is Boltzmann's constant, and $T$ is the receiver system temperature in kelvins. Because the UAV flies at the height of $h=100~ \mathrm{m}$, small changes of receiver's temperature between $h=0~ \mathrm{m}$ and $h=100 ~\mathrm{m}$ have a small impact on $N_0$. Hence, we assume same $N_0$ for aerial and terrestrial nodes.}. $\bar{P}_{s}$ and $\bar{P}_{r}$ indicate the average transmit power at each time slot at the source and relay nodes, respectively. During $N$ time slots, the total transmit energy for the source and relay nodes must be less than $E_s=N\bar{P}_{s}$ and $E_r=N\bar{P}_{r}$, respectively \cite{zeng2016throughput}. Hence, the BS and UAV can consume different powers at different time slots conditioned upon their energy consumption constraint.
  Finally, we assume that the relay node operates
in frequency division duplexing (FDD) mode with equal bandwidth allocated for information reception from the BS and transmission to the vehicles.

\section{DYNAMIC NOMA/OMA SCHEME}
In the NOMA scheme, the BS sends the combination of vehicle messages to the relay node based on superposition coding (SC) as $s_3[n]=\sqrt{ p_{1}[n]} s_{1}[n]+\sqrt{p_{2}[n]} s_{2}[n]$,
%This paper investigates the case that there is two vehicles in just one subcarrier and their messages are distinguished by their power levels (i.e., single-carrier PD-NOMA). The case that there is multiple sub-carriers in the system  (multiple carrier PD-NOMA) is out of this paper's scope. In that case, we  must apply vehicle pairing algorithms to pair the vehicles with each other properly and send their messages at the same subcarrier.\par
where $p_{1}[n]$ and $p_{2}[n]$ are power allocation coefficients at time slot $n$ for each vehicle, and we call them NOMA coefficients. $s_{1}[n]$ and $s_{2}[n]$ are the messages of vehicle 1 and vehicle 2, respectively. We assume that $E\{|s_{1}[n]|^{2}\}=E\{|s_{2}[n]|^{2}\}=1$. In NOMA, less power is allocated to the strong vehicle, and more power is allocated to the weak vehicle \cite{saito2013non}. Then we can perform SIC to remove far vehicle interference from the received signal at the strong vehicle. 
%Therefore, power allocation coefficients for these vehicles must satisfy $a_{1}[n] <a_{2}[n]$ and $ a_{1}[n]+a_{2}[n]=1 $.\par
Note that due to assuming the free space path loss model for air to ground channels, we only need to sort NOMA vehicles based on their distances. 
%This is in contrast to terrestrial links with multipath fading where we must sort the vehicles based on their random channel powers. 
As we can see in Fig. \ref{fig:model}, based on the trajectory of the UAV and the locations of vehicles, vehicle 1 (vehicle 2) can be near vehicle in some time slots. Hence, when vehicle 1 (vehicle 2) is the near vehicle, we apply SIC at vehicle 1 (vehicle 2) to remove the interference of vehicle 2 (vehicle 1). The received signal at the relay node at time slot $n$ is given by
$y_{r}[n]=\sqrt{h_{r}[n]p_{1}[n]}s_{1}[n]+\sqrt{h_{r}[n]p_{2}[n]}s_{2}[n]+z_{r}[n]$.
 The relay works in the AF mode in our system and  amplifies the received signal at each time slot with amplification gain $\rho[n]$. According to \cite{patel2006statistical}, and assuming that  we allocate power $p_{r}[n]$ for the UAV, we can write the amplification gain as $\rho[n]=\frac{p_{r}[n]}{(p_{1}[n]+p_{2}[n])h_{r}[n]+\sigma^{2}}$. Then, the relay transmits the amplified signal $\sqrt{\rho[n]}y_{r}[n]$ to the vehicles. The received signal at each vehicle is given by
\begin{equation} \label{ykn}
\begin{split}
y_{k}[n]&=\sqrt{\rho[n]h_{rk}[n]}y_{r}[n]+z_{k}[n]\\&=\sqrt{\rho[n]h_{rk}[n]h_{r}[n]p_{1}[n]}s_{1}[n]\\&+\sqrt{\rho[n]h_{rk}[n]h_{r}[n]p_{2}[n]}s_{2}[n]\\&+\sqrt{\rho[n]h_{rk}[n]}z_{r}[n]+z_{k}[n], ~~~k=1,2.
\end{split}
\end{equation}
%where $z_{k}[n]$ ($k=1,2$) is AWGN noise at the vehicles at time slot $n$. We assume that, this noise has the same variance $\sigma^{2}$ at each vehicle for all time slots. \par

\begin{figure}
    \centering
  \includegraphics[width=\linewidth]{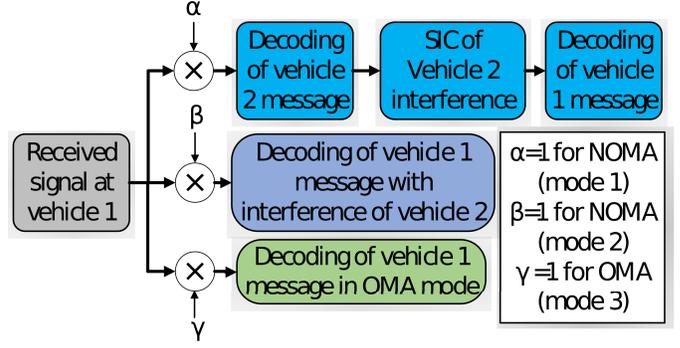}
  \caption{Decoding structure at vehicle 1's receiver with proposed dynamic NOMA/OMA.}
  \label{fig:dynamic}
\end{figure}

One can see from (\ref{ykn}) that for the case where vehicle 1 (vehicle 2) is the near vehicle, in order to perform SIC, we have to allocate more power for vehicle 2 (vehicle 1).  Note that due to the mobility of the vehicles and the aerial relay in our system, the channel power gain between nodes changes a lot. Hence, in contrast to the scenarios with fixed users in which the role of the strong user is given to the user closest to the BS, the roles of strong and weak users alternate between the two vehicles in different time slots for our system. This means that the vehicle which must perform SIC alternates in different time slots, and consequently, the complexity of decoding is divided between vehicles. Also, we will prove that in some cases NOMA is not superior to OMA. Applying OMA in these cases reduces the complexity of performing SIC at the vehicles in the cost of a minor decrease in throughput. Based on these observations, we propose a dynamic NOMA/OMA scheme with the following three modes for transmission of vehicle messages: SIC at vehicle 1, SIC at vehicle 2, and OMA which are indicated by mode 1 ($m=1$), mode 2 ($m=2$), and mode 3 ($m=3$), respectively. The structure of the receiver's decoder at vehicle 1 for the proposed dynamic NOMA/OMA scheme has been shown in Fig. \ref{fig:dynamic}. Note that the same structure can be depicted for the decoder of vehicle 2.
  \par 
\subsection{Mode 1: SIC at Vehicle 1}

 In this case, vehicle 1 is the near vehicle and performs SIC to remove the interference of vehicle 2. Hereafter, we name the multiple access scheme of this case `mode 1' ($m=1$). At vehicle 1, message of vehicle 2 is decoded first. Then, this decoded message is subtracted from the received signal at vehicle 1. Finally, vehicle 1 decodes its own message without any interference. Hence, the effective SINR of the vehicle 2 message observed at vehicle 1 equals
$\gamma_{1,2}[n]=\frac{\rho[n]h_{1}[n]h_{r}[n]p_{2}[n]}{\rho[n]h_{1}[n]h_{r}[n]p_{1}[n]+\rho[n]h_{1}[n]\sigma^{2}+\sigma^{2}}$, and the achievable rate of decoding the vehicle 2 message at vehicle 1 equals $R_{1,2}[n]=\log_{2}(1+\gamma_{1,2}[n])$. Then, the achievable rate of vehicle 1 equals
\begin{equation}
\begin{split}
\label{r1-mode1}
R_{1}^{1}[n]&=\log_{2}(1+\\&\frac{p_{r}[n]h_{1}[n]h_{r}[n]p_{1}[n]}{p_{r}[n]h_{1}[n]\sigma^{2}+(p_{1}[n]+p_{2}[n])h_{r}[n]\sigma^{2}+\sigma^{4}}),
\end{split}
\end{equation}
in which $R_{k}^{m}[n]$ indicates the achievable rate of vehicle $k$ at mode $m$. From (\ref{ykn}), the achievable rate of the vehicle 2 message at vehicle 2 equals $R_{2,2}[n]=\log_{2}(1+\frac{\rho[n]h_{2}[n]h_{r}[n]p_{2}[n]}{\rho[n]h_{2}[n]h_{r}[n]p_{1}[n]+\rho[n]h_{2}[n]\sigma^{2}+\sigma^{2}})$. On the other hand, the vehicle 2 message must be decoded at vehicle 1 in order to perform SIC. Hence, the achievable rate of vehicle 2 equals $R_{2}[n]=\min(R_{1,2}[n],R_{2,2}[n])$.
In this mode, we suppose that vehicle 1 is closer to the UAV in comparison to vehicle 2. Hence, we write $h_{1}[n]>h_{2}[n]$. Therefore, $R_{1,2}[n]>R_{2,2}[n]$, and we can write \eqref{r2-mode1} at the top of the next page.

\begin{figure*}[!t]
\normalsize
\begin{align}\label{r2-mode1}
    R_{2}^{1}[n]=\log_{2}(1+\frac{p_{r}[n]h_{2}[n]h_{r}[n]p_{2}[n]}{p_{r}[n]h_{2}[n]h_{r}[n]p_{1}[n]+p_{r}[n]h_{2}[n]\sigma^{2}+(p_{1}[n]+p_{2}[n])h_{r}[n]\sigma^{2}+\sigma^{4}}).
\end{align}
\hrulefill \vspace*{0pt}
\end{figure*}

\subsection{Mode 2: SIC at Vehicle 2}

In this mode, vehicle 2 is the near vehicle, and performs SIC to remove the interference of vehicle 1. The formulation of achievable rates for both vehicles in this mode is similar to the previous case except that we perform SIC at vehicle 2. Hence, we only write the achievable rates of vehicles in this section. The achievable rate of vehicle 2 and vehicle 1 are given by
\begin{equation}
\label{r2-mode2}
\begin{split}
R^{2}_{2}[n]&=\log_{2}(1+\\&\frac{p_{r}[n]h_{2}[n]h_{r}[n]p_{2}[n]}{p_{r}[n]h_{2}[n]\sigma^{2}+(p_{1}[n]+p_{2}[n])h_{r}[n]\sigma^{2}+\sigma^{4}}),
\end{split}
\end{equation}
and \eqref{r1-mode2} at the top of the next page. 
\begin{figure*}[!t]
\normalsize
\begin{align}\label{r1-mode2}
    R_{1}^{2}[n]=\log_{2}(1+\frac{p_{r}[n]h_{1}[n]h_{r}[n]p_{1}[n]}{p_{r}[n]h_{1}[n]h_{r}[n]p_{2}[n]+p_{r}[n]h_{1}[n]\sigma^{2}+(p_{1}[n]+p_{2}[n])h_{r}[n]\sigma^{2}+\sigma^{4}}).
\end{align}
\hrulefill \vspace*{0pt}
\end{figure*}

\subsection{Mode 3: OMA Scheme}

    We apply frequency division multiple access (FDMA) for two vehicles in this mode.  We assume that the bandwidth of each time slot $n$ is equally divided between two vehicles at both the BS and relay nodes and that the vehicles send their messages in an orthogonal manner. Hereafter, we name this case `mode 3' (m=3). We assume that different transmit powers $p_{1}[n]$ and $p_{2}[n]$ at the BS are allocated to vehicle 1 and vehicle 2, respectively. The BS sends $\sqrt{ p_{1}[n]}\ s_{1}[n]$, and $\sqrt{p_{2}[n]}\ s_{2}[n]$, at two orthogonal frequency bands of each time slot. The received signals at the relay node equals $y_{r}^k[n]=\sqrt{h_{r}[n]p_{k}[n]}s_{k}[n]+z_{r}[n]$ for $k=1,2$. Note that due to dividing frequency bands for the two vehicles, the variance of the noise $z_{r}[n]$ is half of the NOMA case. Hence, we assume that this noise has the same variance $\sigma^{2}_{O}=0.5\sigma^{2}$ in all time slots. At the  relay node, we assume that the power of relay $p_{r}[n]$ is equally divided between two vehicles. Therefore, we can write the amplification gain of each vehicle message as 
$\rho_{k}[n]=\frac{0.5p_{r}[n]}{p_{k}[n]h_{r}[n]+\sigma^{2}_{O}}$ for $k=1,2$.
Then, the relay transmits the amplified signal $\sqrt{\rho_{k}[n]}y_{r}^k[n]$ to the vehicles. The received signal at each vehicle is given by
\begin{equation} \label{ykn_OMA}
\begin{split}
y_{k}[n]&=\sqrt{\rho_{k}[n]h_{k}[n]}y_{r}^k[n]+z_{k}[n]\\&=\sqrt{\rho_{k}[n]h_{k}[n]}(\sqrt{h_{r}[n]p_{k}[n]}s_{k}[n]+z_{r}[n])+z_{k}[n]\\
&=\sqrt{\rho_{k}[n]h_{k}[n]h_{r}[n]p_{k}[n]}s_{k}[n]\\&+\sqrt{\rho_{k}[n]h_{k}[n]}z_{r}[n]+z_{k}[n], ~~~k=1,~2.
\end{split}
\end{equation}
 We assume that $z_{k}[n]$ has the same variance $\sigma^{2}_{O}=0.5\sigma^{2}$ at each vehicle. From (\ref{ykn_OMA}), the effective SINR of each vehicle equals $\gamma_{k}[n]=\frac{\rho_{k}[n]h_{k}[n]h_{r}[n]p_{k}[n]}{\rho_{k}[n]h_{k}[n]\sigma^{2}_{O}+\sigma^{2}_{O}}$,
and the rates of vehicles are given by
\begin{equation}
\label{rate-oma}
\begin{split}
R_{k}^{3}[n]&=\frac{1}{2}\log_{2}(1+\frac{p_{r}[n]h_{k}[n]h_{r}[n]p_{k}[n]}{p_{r}[n]h_{k}[n]\sigma^{2}_{O}+2p_{k}[n]h_{r}[n]\sigma^{2}_{O}+2\sigma^{4}_{O}}),
\end{split}
\end{equation}
for $k=1,2$.
\subsection{Proposed Dynamic NOMA/OMA Scheme}
In this section, we first introduce two propositions which compare the sum-rate and min-rate performance of NOMA and OMA at the high SNR regime. We assume that $R_{i,j}^{\infty}$ shows the approximated value for rate at the high SNR regime, where $i\in\{S,M\}$ differentiates sum-rate and min-rate, and $j\in\{O,N\}$ differentiates OMA and NOMA schemes.
\begin{figure}[t]
\begin{subfigure}{0.5\textwidth}
\includegraphics[width=\textwidth]{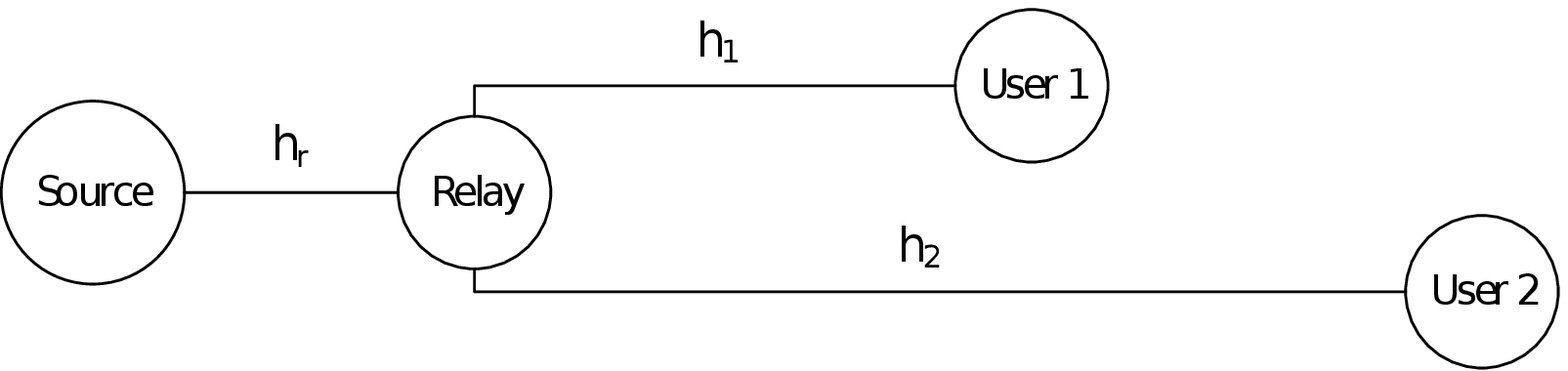} 
\caption{Case 1: $h_{r}>h_{1}>h_{2}$}
\label{fig:case1}
\end{subfigure}
\begin{subfigure}{0.5\textwidth}
\includegraphics[width=\textwidth]{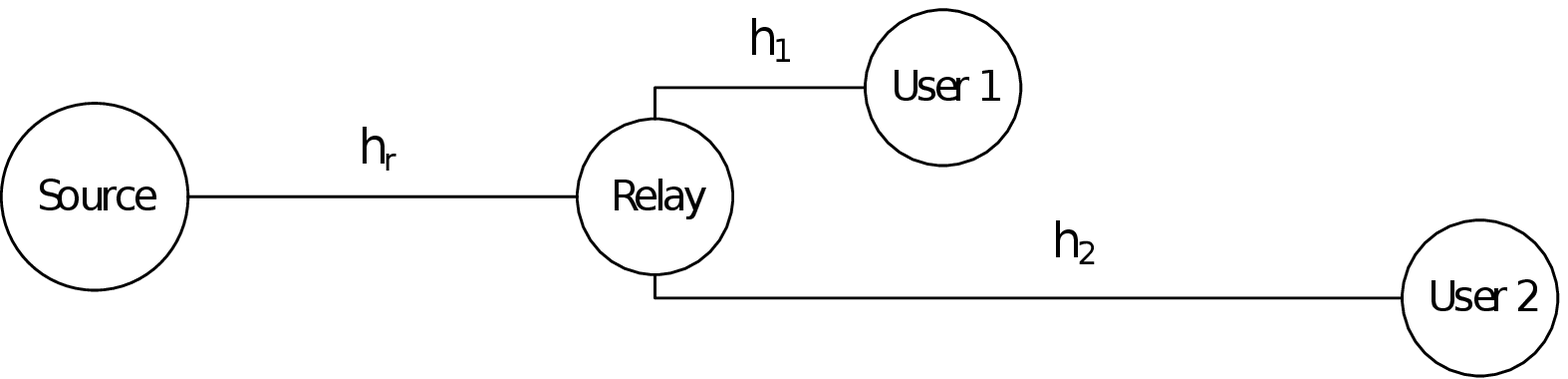}
\caption{Case 2: $h_{1}>h_{r}>h_{2}$}
\label{fig:case2}
\end{subfigure}
\caption{Two cases where NOMA has better sum-rate than OMA (refer to Proposition \ref{proposition-sum}).}
\label{fig:cases12}
\end{figure}
\begin{proposition}\label{proposition-sum}
In a two-user network with a dedicated AF relay and in the absence of a direct link between the BS and two users, NOMA always has better or equal sum-rate performance in comparison to OMA at high SNR regime.
%In a two-user network with a dedicated relay and in the absence of direct link between BS and the two users, NOMA  has superior sum-rate performance in comparison to OMA at the high SNR in two cases (Fig. \ref{fig:cases12}): Case 1) The BS to relay link is stronger than the relay to user links, and the links between the relay and users are degraded. Case 2) The link between the BS and the relay is better than the link between the relay and the weak user.
\end{proposition}
\begin{proof}
We use the same notations as in the system model section except that we take the transmit SNRs at the BS and relay nodes (i.e., $\rho=\frac{\bar{P}_{s}}{\sigma^{2}}=\frac{\bar{P}_{r}}{\sigma^{2}}=\frac{P}{\sigma^{2}}$) into channel power gains. Hence, new channel power gains are $h_{i}^{\infty}=\rho h_{i}$ ($i\in\{r,1,2\}$) where the superscript $\infty$ indicates the high SNR regime. Also, with this new notation for channel power gain, we have $p_{1}+p_{2}=1$ and $p_{r}=1$.  Assuming that $h_{1}>h_{2}$, we have to perform SIC at user 1. Hence, utilizing the achievable rates of users at mode 1 from (\ref{r1-mode1}) and (\ref{r2-mode1}) we have
\begin{equation}
\label{noma-sumrate}
\begin{split}
R_{S,N}^{\infty}&=R_{1}^{1,\infty}+R_{2}^{1,\infty}
=\log_{2}(1+\frac{p_{r}h_{1}^{\infty}h_{r}^{\infty}p_{1}}{p_{r}h_{1}^{\infty}+(p_{1}+p_{2})h_{r}^{\infty}+1})\\&+\log_{2}(1+\frac{p_{r}h_{2}^{\infty}h_{r}^{\infty}p_{2}}{p_{r}h_{2}^{\infty}h_{r}^{\infty}p_{1}+p_{r}h_{2}^{\infty}+(p_{1}+p_{2})h_{r}^{\infty}+1})\\&\approx\log_{2}(\frac{h_{1}^{\infty}h_{r}^{\infty}}{h_{1}^{\infty}+h_{r}^{\infty}})+\log_{2}(p_1)+\log_{2}(1+\frac{p_{2}}{p_{1}})\\&=\log_{2}(\frac{h_{1}^{\infty}h_{r}^{\infty}}{h_{1}^{\infty}+h_{r}^{\infty}})\approx\log_2\bigg(\min(h_{r}^{\infty},h_{1}^{\infty})\bigg).
\end{split}
\end{equation}
On the other hand, for the OMA scheme, we use the rates of the users at mode 3 from (\ref{rate-oma}) as
\begin{equation}
\label{oma-sumrate}
\begin{split}
R_{S,O}^{\infty}&=\sum_{k=1}^2R_{k}^{3,\infty}=\sum_{k=1}^2\frac{1}{2}\log_{2}(1+\frac{2p_{r}h_{k}^{\infty}h_{r}^{\infty}p_{k}}{p_{r}h_{k}^{\infty}+2p_{k}h_{r}^{\infty}+1})
\\&\approx\sum_{k=1}^2\frac{1}{2}\log_2\bigg(\min(h_{r}^{\infty},h_{k}^{\infty})\bigg),
\end{split}
\end{equation}
where we considered $p_1=p_2=0.5$ in (\ref{oma-sumrate}). We can see from (\ref{noma-sumrate}) and (\ref{oma-sumrate}) that with $h_{1}^{\infty}>h_{2}^{\infty}$,  NOMA always has equal or better sum-rate performance in comparison to the OMA scheme. Actually, based on the channel power gains of links, three cases may occur for sum-rate:\\
\textbf{Case 1}: When we have  $h_{r}^{\infty}>h_{1}^{\infty}>h_{2}^{\infty}$ (Fig. \ref{fig:case1}), the sum-rates of NOMA and OMA are $\log_2(h_{1}^{\infty})$ and $\frac{1}{2}\log_2(h_{1}^{\infty})+\frac{1}{2}\log_2(h_{2}^{\infty})$, respectively.  NOMA  has better performance than OMA in this case as
    $R_{S,N}^{\infty}-R_{S,O}^{\infty}\approx\log_2(h_{1}^{\infty})-\frac{1}{2}\log_2(h_{1}^{\infty})-\frac{1}{2}\log_2(h_{2}^{\infty})=\frac{1}{2}\log_2(\frac{h_{1}^{\infty}}{h_{2}^{\infty}})$.
Hence, by increasing the difference between  $h_{1}^{\infty}$ and $h_{2}^{\infty}$, the superiority of NOMA increases.\\
\textbf{Case 2}: When we have $h_{1}^{\infty}>h_{r}^{\infty}>h_{2}^{\infty}$ (Fig. \ref{fig:case2}), the sum-rates of NOMA and OMA are $\log_2(h_{r}^{\infty})$ and $\frac{1}{2}\log_2(h_{r}^{\infty})+\frac{1}{2}\log_2(h_{2}^{\infty})$, respectively. One can see that NOMA has better performance than OMA as $R_{S,N}^{\infty}-R_{S,O}^{\infty}\approx\log_2(h_{r}^{\infty})-\frac{1}{2}\log_2(h_{r}^{\infty})-\frac{1}{2}\log_2(h_{2}^{\infty})=\frac{1}{2}\log_2(\frac{h_{r}^{\infty}}{h_{2}^{\infty}})$.
We see that by increasing the difference between  $h_{r}^{\infty}$ and $h_{2}^{\infty}$, the superiority of NOMA increases.\\
\textbf{Case 3}: When we have  $h_{1}^{\infty}>h_{2}^{\infty}>h_{r}^{\infty}$, the sum-rates of both schemes equals $\log_2(h_{2}^{\infty})$. Hence, NOMA has the same performance as OMA for the case where the BS-to-relay link is worse than the relay-to-user links. Consequently, the proof is completed.
\end{proof}
\begin{remark}\label{remark1}
Proposition \ref{proposition-sum} proves that, for the cases where two users have similar links, or where BS-to-relay link is similar with relay-to-weak user link, applying NOMA has negligible sum-rate gain. Also, when the BS-to-relay link is weak, NOMA and OMA have similar sum-rate. 
%Proposition \ref{proposition-sum} proves that for cases 1 and 2 NOMA has better performance than OMA, and this superiority depends on the difference between channels power gains. When these differences are low, we do not attain much gain from NOMA in comparison to OMA. Hence, due to decoding complexity of SIC at NOMA scheme, it's better to apply OMA scheme. Also, in case 3, NOMA and OMA have similar performance, and hence, we apply OMA scheme.
\end{remark}
\begin{proposition}\label{proposition-min}
In a two-user network with a dedicated AF relay and in the absence of a direct link between the BS and two users, OMA  has superior minimum-rate performance in comparison to NOMA at the high SNR regime.
\end{proposition}
\begin{proof}
With the same notations used in the proof of Proposition \ref{proposition-sum}, the minimum achievable rates of users at high SNRs in mode 1 can be written from (\ref{r1-mode1}), (\ref{r2-mode1}) and (\ref{noma-sumrate})  as
\begin{equation}
\label{noma-minrate}
\begin{split}
R_{M,N}^{\infty}&=\min(R_{1}^{1,\infty},R_{2}^{1,\infty})\\&\approx\min\bigg(\log_{2}(\frac{h_{1}^{\infty}h_{r}^{\infty}}{h_{1}^{\infty}+h_{r}^{\infty}}),\log_{2}(\frac{p_{2}}{p_{1}})\bigg)=\log_2(\frac{p_{2}}{p_{1}}).
\end{split}
\end{equation}
For the OMA scheme we use the rates of the users at mode 3 from (\ref{rate-oma}) and (\ref{oma-sumrate}) as
\begin{equation}
\label{oma-minrate}
\begin{split}
R_{M,O}^{\infty}=\min(R_{1}^{3,\infty},R_{2}^{3,\infty})=R_{2}^{3,\infty}\approx\frac{1}{2}\log_{2}(\frac{h_{2}^{\infty}h_{r}^{\infty}}{h_{2}^{\infty}+h_{r}^{\infty}})
%\approx \frac{1}{2}\log_2\bigg(\min(h_{r}^{\infty},h_{2}^{\infty})\bigg).
\end{split}
\end{equation}
 We can see from (\ref{noma-minrate}) and (\ref{oma-minrate}) that the OMA scheme always has better min-rate performance in comparison to the NOMA scheme at the high SNR regime. Hence, the proof is completed.
\end{proof}
\begin{remark}
Note that the results in Propositions \ref{proposition-sum} and \ref{proposition-min} are not limited to an aerial AF relay; rather, they are general results for any two-user network with a dedicated AF relay. 
\end{remark}
%\begin{remark}
%Due to mobile nature of vehicles and aerial relay, variation in channels' orders is more probable in our system rather than the case that users and relay are fixed.
%\end{remark}  
In Fig. \ref{fig:prop12-approx}, we can see the approximated and exact rates for sum-rate and min-rate of OMA and NOMA schemes at Proposition \ref{proposition-sum} and Proposition \ref{proposition-min}. In Fig. \ref{fig:oma-approx}, we can see that when $P>22 ~\mathrm{dBm}$,  we have $|R_{S,O}^{\infty}-R_{S,O}|\leq 0.2$ and $|R_{M,O}^{\infty}-R_{M,O}|\leq 0.2$ for the approximations (\ref{oma-sumrate}) and (\ref{oma-minrate}) at sum-rate and min-rate of OMA scheme, respectively. Also, we can see in Fig. \ref{fig:noma-approx} that when $P>22 ~\mathrm{dBm}$, we have $|R_{S,N}^{\infty}-R_{S,N}|\leq 0.2$ and $|R_{M,N}^{\infty}-R_{M,N}|\leq 0.2$ for the approximations (\ref{noma-sumrate}) and (\ref{noma-minrate}) at sum-rate and min-rate of NOMA scheme, respectively. Note that for Fig. \ref{fig:prop12-approx}, we have considered default simulation parameters as in Section VI except that $(x_s,y_s)=(250,250)$ and $(y_1,y_2)=(400,500)$.\par
\begin{figure}[t]
\begin{subfigure}{0.5\textwidth}
\includegraphics[ width=\linewidth]{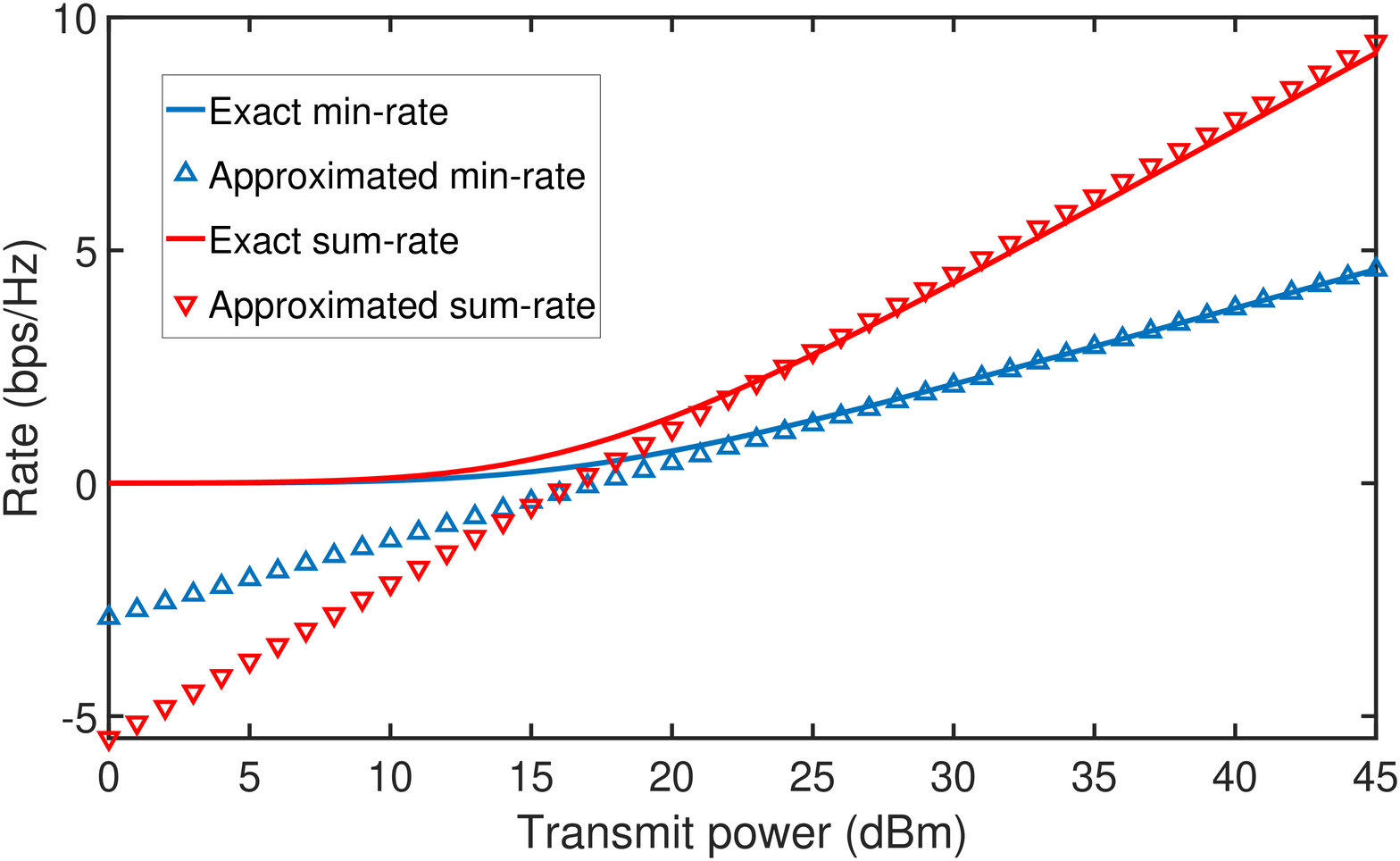} 
\caption{OMA scheme.}
\label{fig:oma-approx}
\end{subfigure}
\begin{subfigure}{0.5\textwidth}
\includegraphics[ width=\linewidth]{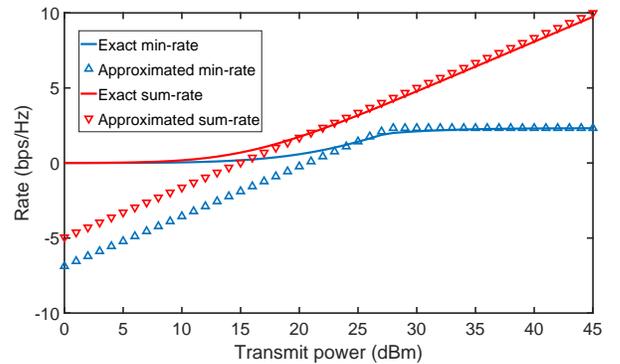}
\caption{NOMA scheme.}
\label{fig:noma-approx}
\end{subfigure}
\caption{Approximated and exact rates for sum-rate and min-rate of OMA and NOMA schemes at Proposition \ref{proposition-sum} and Proposition \ref{proposition-min}.} 
\label{fig:prop12-approx}
\end{figure}
NOMA scheme requires more complex receiver rather than OMA (Fig. \ref{fig:dynamic}). Actually, performing SIC at near user leads to a more complex decoder at the receiver of near user \cite{benjebbour_SIC}, and consequently,  more required power. 
%Also, performing SIC means that receiver requires more processing time to decode the received signal. This leads to high decoding latency  at near user. 
Hence, more power consumption and complexity are the costs that we pay at receiver side to have better throughput by applying the NOMA scheme. 
%On the other hand, in order to perform NOMA in our system, the BS have to perform SC to transmit messages of users, and hence, different transmission schemes are applied at BS for OMA and NOMA schemes. Furthermore, the structure of relaying scheme in relay node is different for NOMA and OMA schemes. \\
 Based on the observations in Proposition \ref{proposition-sum} and Remark \ref{remark1}, we propose a dynamic NOMA/OMA scheme in which OMA is selected when applying NOMA has zero or only negligible gain. 
 % It is clear that when OMA and NOMA schemes have the same sum-rate performance, due to above-mentioned costs for NOMA scheme at transceiver design, our proposed dynamics scheme selects OMA mode. Note that according to Proposition \ref{proposition-sum}, this case happens when the BS-to-relay link is worse than the relay-to-user links (case 3 in Proposition \ref{proposition-sum}), or two users have similar links (case 1 in Proposition \ref{proposition-sum}), or BS-to-relay link is similar with relay-to-weak user link (case 2 in Proposition \ref{proposition-sum}). 
According to Proposition \ref{proposition-sum}, the gain of NOMA scheme over OMA relies on the channel gains of BS-to-relay and relay-to-users links. Due to movements of UAV node and vehicles in our system model, these channel gains are changing fast, and hence, their orders change quickly. Hence, as one can see in Fig. \ref{fig:dynamic}, transceiver parts of nodes have to quickly switch among three working modes, i.e., OMA, NOMA with SIC at user 1 and NOMA with SIC at user 2. These quick switches causes more power and complexity costs for transceiver of BS, relay and users.  In order to prevent our dynamic scheme from the quick switches among working modes and keep the complexity costs of applying NOMA scheme as little as possible, we define a superiority threshold parameter $R_{th}^S$. When the sum-rate superiority of NOMA in comparison to OMA is less than this threshold value, i.e., $R_{S,N}-R_{S,O}<R_{\text{th}}^{S}$, our dynamic scheme is set to OMA mode. $R_{th}^S$ can be chosen based on the priorities of the network operator between complexity and throughput. 
%If transceiver complexity, consumed power at users and latency are less important in a network, we can choose $R_{th}^S=0$ to reach better throughput. On the hand, if transceiver complexity, consumed power at users and latency are important in a network, we can choose a non-zero value for $R_{th}^S$. 
It is clear that increasing $R_{th}^S$ leads to selection of more OMA working modes by proposed dynamic scheme which leads to less decoding and switching complexity in transceiver design in the cost of missing more throughput.\par
   Note that in Proposition \ref{proposition-sum}, we assumed that $h_{1}>h_{2}$, and we have three other states for the channel gain orders in the case where $h_{2}>h_{1}$.  At each state, based on the order and difference of channel gains, we can apply NOMA or OMA. At time slot $n$ and state $s$, in order to determine three modes for multiple access scheme, i.e., $m=1,2,3$, we introduce three binary matrices, namely $\mathbf{A}=[\alpha_{s,n}]_{S\times N}$, $\mathbf{B}=[\beta_{s,n}]_{S\times N}$, and $\mathbf{\Gamma}=[\gamma_{s,n}]_{S\times N}$, respectively. The summary of the proposed dynamic multiple access scheme is shown in Table \ref{table:1}.
\begin{table}[t]
\tiny
%\footnotesize
\centering
%\begin{center}
\caption{Proposed dynamic NOMA/OMA scheme.}
\tabcolsep=0.11cm
%\begin{singlespace}
\scalebox{1}{
 \begin{tabular}{|c| c| c| c|c|c|c|} 
 \hline
 State & Channel gain order& Channel difference& Multiple access scheme& $\alpha_{s,n}$ &$\beta_{s,n}$ &$\gamma_{s,n}$ \\ [0.5ex]
 \hline\hline
 %\multirow{2}{1.5em}{s=1 s=2} &
 s=1 &
 \multirow{2}{9em}{$h_{r}>h_{1}>h_{2}$}& $\frac{1}{2}\log_2(\frac{h_{1}}{h_{2}})>R_{\text{th}}^{S}$ & NOMA- SIC at vehicle 1 (m=1)& 1 &0&0\\\cline{1-1}\cline{3-7} 
 s=2& &$\frac{1}{2}\log_2(\frac{h_{1}}{h_{2}})<R_{\text{th}}^{S}$ & OMA (m=3) & 0 & 0 & 1\\ 
\hline
s=3 & \multirow{2}{9em}{$h_{1}>h_{r}>h_{2}$}& $\frac{1}{2}\log_2(\frac{h_{r}}{h_{2}})>R_{\text{th}}^{S}$ & NOMA- SIC at vehicle 1 (m=1) & 1 &0&0\\\cline{1-1}\cline{3-7}  
s=4 & &$\frac{1}{2}\log_2(\frac{h_{r}}{h_{2}})<R_{\text{th}}^{S}$ & OMA (m=3)& 0 & 0 & 1\\ 
\hline
s=5 & $h_{1}>h_{2}>h_{r}$& -& OMA (m=3) & 0&0&1\\ 
\hline
s=6 & \multirow{2}{9em}{$h_{r}>h_{2}>h_{1}$}& $\frac{1}{2}\log_2(\frac{h_{2}}{h_{1}})>R_{\text{th}}^{S}$ & NOMA- SIC at vehicle 2 (m=2) & 0 &1&0\\ \cline{1-1}\cline{3-7}  
 s=7& &$\frac{1}{2}\log_2(\frac{h_{2}}{h_{1}})<R_{\text{th}}^{S}$ & OMA (m=3) & 0 & 0 & 1\\ 
\hline
s=8 & \multirow{2}{9em}{$h_{2}>h_{r}>h_{1}$}& $\frac{1}{2}\log_2(\frac{h_{r}}{h_{1}})>R_{\text{th}}^{S}$ & NOMA- SIC at vehicle 2 (m=2) & 0 &1&0\\\cline{1-1}\cline{3-7}   
 s=9& &$\frac{1}{2}\log_2(\frac{h_{r}}{h_{1}})<R_{\text{th}}^{S}$ & OMA (m=3) & 0 & 0 & 1\\ 
\hline
s=10& $h_{2}>h_{1}>h_{r}$& -& OMA (m=3) & 0&0&1\\ 
\hline
 % s=6 & 88 & 788 & 6344 \\ [1ex] 
\end{tabular}
}
%\end{singlespace}
\label{table:1}
%\end{center}
\end{table}
\section{Sum-Rate Maximization}
In this section, we formulate the optimization problem of our system. We aim to optimize the end-to-end sum-rate by finding the optimal values for relay trajectory ($\mathbf{X}=[x_{r}[1],...,x_{r}[N]]$ and $\mathbf{Y}=[y_{r}[1],...,y_{r}[N]]$), power coefficients of NOMA vehicles at the BS ($\mathbf{P_{1}}=[p_{1}[1],...,p_{1}[N]]$ and $\mathbf{P_{2}}=[p_{2}[1],...,p_{2}[N]]$), allocated power of the relay node ($\mathbf{P_{r}}=[p_{r}[1],...,p_{r}[N]]$), and operation mode selection matrices for multiple access scheme ($\mathbf{A}=[\alpha_{s,n}]_{S\times N}$, $\mathbf{B}=[\beta_{s,n}]_{S\times N}$, and $\mathbf{\Gamma}=[\gamma_{s,n}]_{S\times N}$). Therefore, the optimization problem can be formulated as
%\begin{subequations}
\begin{alignat}{2}
&\text{(P1):}&&\underset{\mathbf{P}_1,\mathbf{P}_2,\mathbf{P}_r,\mathbf{X},\mathbf{Y},\mathbf{A},\mathbf{B},\mathbf{\Gamma}}{\text{max}}~~~        \sum_{n=1}^{N}\sum_{s=1}^{S}\sum_{k=1}^{K} \alpha_{s,n}R_{k}^{1}[n]\nonumber\\
&&&~~~~~~~~~~~~~~~~~~~~~~~~~~~~~+\beta_{s,n}R_{k}^{2}[n]+\gamma_{s,n}R_{k}^{3}[n]\label{eq:obj_fun}\\
&\text{s.t.} &      &  \sum_{s=1}^{S} \alpha_{s,n}R_{k}^{1}[n]+\beta_{s,n}R_{k}^{2}[n]+\gamma_{s,n}R_{k}^{3}[n]\geq R_{k}^{t}, ~~\forall n, \forall k, \label{eq:constraint-rk}\\
&&&(x[n]-x[n-1])^{2}+(y[n]-y[n-1])^{2}\leq (V\tau)^{2},\nonumber\\
&&&~~~~~~~~~~~~~~~~~~~~~~~~~~~~~~~~~~~~~~~~~~~n=2, ..., N,\label{eq:constraint-v}\\
&&& (x[1]-x_{s})^{2}+(y[1]-y_s)^{2}\leq(V\tau)^{2},\label{eq:constraint-s}\\
&&& (x[N]-x_{f})^{2}+(y[N]-y_{f})^{2}\leq (V\tau)^{2},\label{eq:constraint-f}\\
              &   &      & x_{\text{min}}\leq x[n]\leq x_{\text{max}},~y_{\text{min}}\leq y[n]\leq y_{\text{max}},~\forall n,\label{eq:constraint-range}\\
              &&&\sum_{n=1}^{N}\sum_{s=1}^{S}(\alpha_{s,n}+\beta_{s,n}+\gamma_{s,n})(p_{1}[n]+p_{2}[n])\leq E_{s},\label{eq:constraint-ps}\\
              &&& \sum_{n=1}^{N}p_{r}[n]\leq E_{r}, \label{eq:constraint-pr}\\
              &&& \sum_{s=1}^{S}\alpha_{s,n}(p_{2}[n]-p_{1}[n])+\beta_{s,n}(p_{1}[n]-p_{2}[n])\geq 0, ~~\forall n, \label{eq:constraint-noma}
              \\
               &&& p_{r}[n]\geq 0,~ p_{1}[n]\geq 0,~ p_{2}[n]\geq 0, ~~\forall n, \label{eq:constraint+}
\end{alignat}
where $R_{1}^{1}$, $R_{2}^{1}$, $R_{1}^{2}$, $R_{2}^{2}$, and $R_{k}^{3}$ ($k \in \{1,2\}$)  was calculated in (\ref{r1-mode1}), (\ref{r2-mode1}), (\ref{r1-mode2}), (\ref{r2-mode2}), and (\ref{rate-oma}), respectively. Note that $R_{k}^{m}$ indicates the rate of vehicle $k$ at mode $m$, and $K=2$ is the number of vehicles. Constraint (\ref{eq:constraint-rk}) guarantees minimum target rate $R_{k}^{t}$  for vehicle $k$ at time slot $n$. Constraint (\ref{eq:constraint-v}) indicates the velocity constraint for the movements of the relay node at each time slot. Constraints  (\ref{eq:constraint-s}) and (\ref{eq:constraint-f}) show that the UAV must start and finish its trajectory at specified locations. Constraint (\ref{eq:constraint-range}) indicates the range that the UAV can fly. Constraints (\ref{eq:constraint-ps}) and (\ref{eq:constraint-pr})   indicate the total energy constraints for the BS and relay  nodes, respectively, in which $E_s=N\bar{P}_s$ and $E_r=N\bar{P}_r$. Constraint (\ref{eq:constraint-noma}) imposes a constraint to the allocated powers of the vehicles, in order to perform SIC at each time slot $n$. Note that $\alpha_{s,n}$, $\beta_{s,n}$, and $\gamma_{s,n}$ are indicator functions for three operation modes, and at each time slot $n$, only one of them must equal one.
 \par
The objective function of (P1) and constraint (\ref{eq:constraint-rk}) are not concave with respect to trajectory and power variables. Also, (P1) contains operation mode coefficients which are integer variables. Hence, (P1) is a NP-hard non-convex optimization problem and can not be solved using standard convex  optimization methods. In order to solve this problem, we use the AO method. Therefore, we first solve trajectory planning and operation mode selection sub-problem with fixed power allocations. Then, we propose a solution for power allocation sub-problem  with fixed trajectory and operation mode coefficients. Finally, by alternatively optimizing the trajectory, mode selection, and power allocation sub-problems, we propose an iterative algorithm  to solve (P1).  
\subsection{Trajectory Planning and Operation Mode Selection for Fixed Power Allocation}
In this subsection, we solve the first sub-problem of (P1) to find the optimal values for the UAV trajectory (i.e., $\mathbf{X}$ and $\mathbf{Y}$) and operation mode coefficients (i.e., $\mathbf{A}$, $\mathbf{B}$, and $\mathbf{\Gamma}$) assuming that the allocated powers are given. Hence, we have
\begin{alignat}{2}
\text{(P1.1):}~~~~&\underset{\mathbf{X},\mathbf{Y},\mathbf{A},\mathbf{B},\mathbf{\Gamma}}{\text{max}}        &~~~~&\sum_{n=1}^{N}\sum_{s=1}^{S}\sum_{k=1}^{K} \alpha_{s,n}R_{k}^{1}[n]\nonumber\\
&&&~~~~~~+\beta_{s,n}R_{k}^{2}[n]+\gamma_{s,n}R_{k}^{3}[n]\label{eq:obj_fun_p1.1}\\
&\text{s.t.}&  &     (\ref{eq:constraint-rk}),  (\ref{eq:constraint-v}), (\ref{eq:constraint-s}), (\ref{eq:constraint-f}), (\ref{eq:constraint-range}). \nonumber
\end{alignat}
%\end{equation}
If we replace the channel model formula in (\ref{channel-model}) at each vehicle rate formula in the equations (\ref{r1-mode1}), (\ref{r2-mode1}), (\ref{r2-mode2}), (\ref{r1-mode2}) and (\ref{rate-oma}), one can see that (P1.1) is not a convex problem with respect to variables $\mathbf{X}$ and $\mathbf{Y}$. Indeed, the objective function in (\ref{eq:obj_fun_p1.1}) and constraint (\ref{eq:constraint-rk}) are not concave functions. Also, note that operation mode coefficients are integer variables. We propose to use the AO method for solving this problem. For fixed trajectory planning, the channel power gains are known and we can use Table I to determine $\mathbf{A}$, $\mathbf{B}$, and $\mathbf{\Gamma}$. For given operation mode coefficients, (P1.1) is still non-convex, and we utilize the SCA method in which a concave approximation of the original objective function is maximized iteratively.
\begin{lemma}\label{lemma1}
Consider $x$ and $y$ as two variables, and $a$, $b$, $c$, and $d$ as constants. The function  $f(x,y)=\frac{1}{ax+by+cxy+d}$,
is a convex function with respect to $x$ and $y$, if and only if $ax+by+cxy+d>0$, and $ab=cd$.
\end{lemma}
\begin{proof}
See Appendix \ref{lemma 1}. 
\end{proof}
  In the following proposition, we use Lemma \ref{lemma1} to find a concave approximation for the rates. Note that the location of the UAV is indicated by $(x^{l}[n],y^{l}[n])$ at the $l^{\text{th}}$ iteration. Also, the lower-bound (LB) approximated  rate and SINR of vehicle $k$ ($k=1,2$) in mode $m$ ($m=1,2,3$) and at iteration $l$ is indicated by $R_{k,\text{lb}}^{l,m}[n]$ and $\gamma_{k,\text{lb}}^{l,m}[n]$, respectively. Note that we remove the time slot number $n$ in the following proposition equations to make the formulas simpler to read.

\begin{proposition}\label{trajectory-iter-lb}
At each time slot $n$, for any given power allocation at the BS and UAV nodes (i.e., $p_{r},~ p_{1}$, and $p_{2}$), and for any trajectory at the $l^{\text{th}}$ iteration (i.e., $x^l,~\text{and}~y^l$), one approximated concave non-decreasing lower-bound of the achievable rate of vehicle $k$ at mode $m$ and the $(l+1)^{\text{th}}$ iteration of the SCA method is
\begin{equation}
\begin{split}\label{rate-traj}
R_{k,\text{lb}}^{l+1,m}&=    c_m\log_{2}(1+\gamma_{k,\text{lb}}^{l,m}+d_{k,r}^{l,m}(\psi_{r}^{l+1}-\psi_{r}^{l})\\&+d_{k,k}^{l,m}(\psi_{k}^{l+1}-\psi_{k}^{l})),k=1,2,m=1,2,3,
\end{split}
\end{equation}
where $c_m=1$ for $m=1,2$, and $c_m=\frac{1}{2}$ for $m=3$. Also, $\psi_{r}^{l}=\frac{\sigma^2((x^{l})^{2}+(y^{l})^{2}+h^{2})}{\beta_{0}}$, $\psi_{i}^{l}=\frac{\sigma^2((x^{l}-x_{i})^{2}+(y^{l}-y_{i})^{2}+h^{2})}{\beta_{0}}$, for $i=1,2$, and SINR of vehicle $k$ (for $k=1,2$) equals
    $\gamma_{k,\text{lb}}^{l,m}= d_mp_{r}p_{k}$, where $d_m^{-1}=p_{r}\psi_{r}^{l}+(p_{1}+p_{2})\psi_{k}^{l}+\psi_{r}^{l}\psi_{k}^{l}+(p_{1}+p_{2})p_{r}$, for $m=1,2$,  and $d_m^{-1}=p_{r}\psi_{r}^{l}+p_{k}\psi_{k}^{l}+\psi_{r}^{l}\psi_{k}^{l}+p_{r}p_{k}$, for $m=3$. Partial derivative of SINR of vehicle $k$ (for $k=1,2$) to variable $i\in\{r,k\}$ is denoted by $d_{k,i}^{l,m}[n]$ and equals
    $d_{k,r}^{l,m}=-d_m^2p_{r}p_{k}(p_{r}+\psi_{k}^{l})$, for $m=1,2,3$,
    $d_{k,k}^{l,m}=-d_m^2p_{r}p_{k}((p_{1}+p_{2})+\psi_{r}^{l})$, for $m=1,2$, and
    $d_{k,k}^{l,m}=-d_m^2p_{r}p_{k}(p_{k}+\psi_{r}^{l})$, for $m=3$.
\end{proposition}
\begin{proof}
See Appendix \ref{appendix-traj}.  
\end{proof}
This proposition shows that for given trajectory and mode coefficients at the $l^{\text{th}}$ iteration, the sum-rate of vehicles at (P1.1) is lower-bounded  by the summation of rates in (\ref{rate-traj}). It then follows that the optimal value of (P1.1) is lower-bounded by the optimal value of the following problem 
%\begin{subequations}
\begin{alignat}{2}
\text{(P1.2):}~~&\underset{\mathbf{X}^{l+1},\mathbf{Y}^{l+1}}{\text{max}}    ~~~    &&\sum_{n=1}^{N}\sum_{s=1}^{S}\sum_{k=1}^{K} \alpha_{s,n}^{l}R_{k,\text{lb}}^{l+1,1}[n]\nonumber\\
&&&~~~~~+\beta_{s,n}^{l}R_{k,\text{lb}}^{l+1,2}[n]+\gamma_{s,n}^{l}R_{k,\text{lb}}^{l+1,3}[n]\label{eq-traj:obj_fun}\\
&\text{s.t.}&     &  \sum_{s=1}^{S} \alpha_{s,n}^{l}R_{k,\text{lb}}^{l+1,1}[n]+\beta_{s,n}^{l}R_{k,\text{lb}}^{l+1,2}[n]\nonumber\\
&&&~~~~~+\gamma_{s,n}^{l}R_{k,\text{lb}}^{l+1,3}[n]\geq R_{k}^{t}, \forall n, \forall k, \label{eq-traj:constraint-r1}\\
&&& (\ref{eq:constraint-v}), (\ref{eq:constraint-s}), (\ref{eq:constraint-f}), (\ref{eq:constraint-range}).\nonumber
\end{alignat}
Problem (P1.2) is a convex problem which can be efficiently solved by convex optimization techniques such as the interior-point method. Note that we proved in Proposition \ref{trajectory-iter-lb} that the objective function of (P1.2) is non-decreasing over iterations and is globally upper-bounded by the optimal value of (P1.1). Therefore, the proposed sub-optimal algorithm converges. One can see the summary of the proposed iterative method for solving (P1.1) in Algorithm 1. 
\begin{algorithm}
\caption{Iterative trajectory optimization and operation mode selection for (P1.1) with fixed power allocation.}\label{alg1}
\begin{algorithmic}[1]
%\For{$t = 1,2,\ldots$}
\State  Initialize the UAV trajectory $\mathbf{X}^{l}$ and $\mathbf{Y}^{l}$, and let $l=0$.
\Repeat
\State Find the binary coefficients $\alpha_{s,n}^l,~\beta_{s,n}^l~ \text{and}~ \gamma_{s,n}^l$, for each $n$ based on the UAV trajectory $\mathbf{X}^{l}$ and $\mathbf{Y}^{l}$ and locations of vehicles and BS (refer to Table \ref{table:1}).
\State Solve the convex problem (P1.2) by the interior-point method and find the optimal solution $\mathbf{X}^{l+1}$ and $\mathbf{Y}^{l+1}$.
\State Update $l=l+1$.
\Until {convergence or a maximum number of iterations is reached.}
%\EndFor
\end{algorithmic}
\end{algorithm}
\subsection{Power Allocation with Fixed Trajectory}
In this subsection, we solve the second sub-problem of (P1) to find optimal values for the power allocation of vehicles  at the BS ($\mathbf{P}_1$ and $\mathbf{P}_2$) and the UAV power ($\mathbf{P}_r$), assuming that the trajectory of the UAV ($\mathbf{X}$ and $\mathbf{Y}$) is fixed. Note that for a fixed trajectory, the operation mode coefficients $\mathbf{A}$, $\mathbf{B}$, and $\mathbf{\Gamma}$ are known. Hence, the optimization problem can be written as
%\begin{subequations}
\begin{alignat}{2}
\text{(P1.3):}~~~&\underset{\mathbf{P_1},\mathbf{P_2},\mathbf{P_r}}{\text{max}}  ~~~~~~      &&\sum_{n=1}^{N}\sum_{s=1}^{S}\sum_{k=1}^{K} \alpha_{s,n}R_{k}^{1}[n]\nonumber\\
&&&~~~~~~~~+\beta_{s,n}R_{k}^{2}[n]+\gamma_{s,n}R_{k}^{3}[n]\label{p14}\\
&\text{s.t.}&       & (\ref{eq:constraint-rk}),  (\ref{eq:constraint-ps}), (\ref{eq:constraint-pr}), (\ref{eq:constraint-noma}), (\ref{eq:constraint+}). \nonumber
\end{alignat}
%\end{equation}
 It is clear that the objective function of (P1.3) in (\ref{p14}) and constraint (\ref{eq:constraint-rk}) are not concave functions with respect to $\mathbf{P}_1$, $\mathbf{P}_2$, and $\mathbf{P}_r$. In the following lemmas, we provide two convex functions which can be used for the concave approximation of vehicle rates.
\begin{lemma}\label{con-log}
Consider $x$, $y$, and $z$ as three variables, and $a$, $b$, $c$, $d$, $e$, and $f$ as constants. Then the function $f(x,y,z)=\log_2(axy+bxz+cx+dy+ez+f)$,
is a concave function with respect to $x$, $y$, and $z$, if $\frac{a}{d}=\frac{b}{e}=\frac{c}{f}$.
\end{lemma}
\begin{proof}
Assuming $\frac{a}{d}=\frac{b}{e}=\frac{c}{f}$, we can show function f as
\begin{equation}
\begin{split}
    f(x,y,z)&=\log_2(axy+bxz+cx+dy+\frac{db}{a}z+\frac{dc}{a})\\&=\log_2(x+\frac{d}{a})+\log_2(ay+bz+c).
 \end{split}
\end{equation}
Hence, $f(x,y,z)$ is the summation of two concave functions, and the proof is completed. 
\end{proof}
\begin{lemma}\label{concave2}
Consider $x$ and $y$ as two variables, and $a$, $b$, $c$, and $d$ as constants. Then the function  
    $f(x,y)=\log_2(axy+bx+cy+d)$
is a concave function with respect to $x$ and $y$, if $\frac{a}{c}=\frac{b}{d}$.
\end{lemma}
\begin{proof}
The proof of this lemma is the same with Lemma \ref{con-log}, and hence we remove it. 
\end{proof}
%In the following lemma, we utilize Lemma \ref{con-log} and Lemma \ref{concave2} to show that the sum-rate of two vehicles can be approximated in the form of a DC function.
%Note that we remove  $n$ in the following lemma and proposition equations to make the formulas simpler to read.
\begin{lemma}\label{DC}
The sum-rate of the two vehicles can be approximated in the form of a DC function as
\begin{equation}
\begin{split}
&R_{\text{dc},1}^{1}+R_{\text{dc},2}^{1}\approx\log_{2}(p_{r}h_{1}h_{r}p_{1}+p_{r}h_{1}+p_{1}h_{r}+1)\\&-\log_2(p_{r}h_{1}+(p_{1}+p_{2})h_{r}+1)\\&+\log_{2}(p_{r}h_{2}h_{r}(p_{1}+p_{2})+p_{r}h_{2}+(p_{1}+p_{2})h_{r}+1)\\&-\log_2(p_{r}h_{2}h_{r}p_{1}+p_{r}h_{2}+p_{1}h_{r}+1),
\label{eq:obj_fun_p_DC1}
\end{split}
\end{equation}
\begin{equation}
\begin{split}
&R_{\text{dc},1}^{2}+R_{\text{dc},2}^{2}\approx\\&\log_{2}(p_{r}h_{1}h_{r}(p_{1}+p_{2})+p_{r}h_{1}+(p_{1}+p_{2})h_{r}+1)\\&-\log_2(p_{r}h_{1}h_{r}p_{2}+p_{r}h_{1}+p_{2}h_{r}+1)\\&+\log_{2}(p_{r}h_{2}h_{r}p_{2}+p_{r}h_{2}+p_{2}h_{r}+1)\\&-\log_2(p_{r}h_{2}+(p_{1}+p_{2})h_{r}+1),
\label{eq:obj_fun_p_DC2}
\end{split}
\end{equation}
\begin{equation}
\begin{split}
\sum_{k=1}^K R_{\text{dc},k}^{3}&=\sum_{k=1}^K \frac{1}{2}\bigg(\log_{2}(p_{r}h_{k}h_{r}p_{k}+p_{r}h_{k}+2p_{k}h_{r}+2)\\&-\log_2(p_{r}h_{k}+2p_{k}h_{r}+2)\bigg),
\label{eq:obj_fun_p_DC3}
\end{split}
\end{equation}
where $R_{\text{dc},k}^{m}$ indicates the DC approximation of vehicle $k$ in mode $m$.
\end{lemma}

\begin{proof}
We can write the sum-rate of the two vehicles at mode $m=1$ from (\ref{r1-mode1}) and (\ref{r2-mode1}) as
  \begin{equation}
\label{dc}
\begin{split}
R_{1}^{1}&+R_{2}^{1}=\log_{2}(\frac{e_1}{f})\\&+\log_{2}(\frac{p_{r}h_{2}h_{r}(p_{1}+p_{2})+p_{r}h_{2}+(p_{1}+p_{2})h_{r}+1}{e_2})\\&=\log_{2}(\frac{e_1}{e_2})\\&+\log_{2}(\frac{p_{r}h_{2}h_{r}(p_{1}+p_{2})+p_{r}h_{2}+(p_{1}+p_{2})h_{r}+1}{f}),
\end{split}
\end{equation}
where $e_k=p_{r}h_{k}h_{r}p_1+p_{r}h_{k}+(p_{1}+p_{2})h_{r}+1$, and $f=p_{r}h_{1}+(p_{1}+p_{2})h_{r}+1$. According to Lemma \ref{con-log}, the second logarithmic function of second equality in (\ref{dc}) is in the form of the difference of two concave functions as (\ref{eq:obj_fun_p_DC1}). For the first logarithmic function at second equality of (\ref{dc}), we subtract $p_2h_{r}$ from the numerator and denominator to approximate it with a DC function as (\ref{eq:obj_fun_p_DC1}). With the same procedure, the sum-rate at $m=2$ can be approximated in the form of a DC function as (\ref{eq:obj_fun_p_DC2}). For $m=3$, from (\ref{rate-oma}) we have $\sum_{k=1}^K R_{k}^{3}=\sum_{k=1}^K \frac{1}{2}\log_{2}(\frac{p_{r}h_{k}h_{r}p_{k}+p_{r}h_{k}+2p_{k}h_{r}+2}{p_{r}h_{k}+2p_{k}h_{r}+2})$,
 which according to Lemma \ref{concave2} is a DC function as (\ref{eq:obj_fun_p_DC3}). Hence, the proof is completed.
\end{proof}

In order to solve (P1.3), we apply the SCA method using the approximated DC sum-rate in Lemma \ref{DC}. To this end, we find a concave approximation of the DC objective function in the following proposition. Note that the power allocations of the BS and UAV nodes are indicated by $(p_r^l,p_1^l,p_2^l)$ and $(p_r^{l+1},p_1^{l+1},p_2^{l+1})$ at the $l^{\text{th}}$ and $(l+1)^{\text{th}}$ iterations, respectively. Also, the lower-bound concave rate of vehicle $k$ at mode $m$ is indicated by $R_{k,\text{lb}}^{l+1,m}$ at each time slot $n$. 
\begin{proposition}\label{sca-power}
At each time slot $n$, for any given UAV trajectory and for any power allocation at the $l^{\text{th}}$ iteration ($p_r^l,p_1^l,p_2^l$), one approximated concave non-decreasing lower-bound of the rate of vehicle $k$ at mode $m$ and at the $(l+1)^{\text{th}}$ iteration of the SCA method equals
\begin{equation}
\begin{split}
R_{k,\text{lb}}^{l+1,m}&=\log_{2}(p_{r}^{l+1}h_{k}h_{r}p_{k}^{l+1}+p_{r}^{l+1}h_{k}+p_{k}^{l+1}h_{r}+1)\\&-\log_2(p_{r}^{l}h_{k}+(p_{1}^{l}+p_{2}^{l})h_{r}+1)\\&-d_{k}^{l,m}(p_1^{l+1}-p_1^{l})-t_{k}^{l,m}(p_2^{l+1}-p_2^{l})-c_{k}^{l,m}(p_r^{l+1}-p_r^{l}),\\&~~~~~~~~~~~~~~~~~~~~~~~~~~(m,k)=(1,1),(m,k)=(2,2),
\label{eq:obj_fun_p_DC_p1}
\end{split}
\end{equation}
\begin{equation}
\begin{split}
&R_{k,\text{lb}}^{l+1,m}=\log_{2}(p_{r}^{l+1}h_{k}h_{r}(p_{1}^{l+1}+p_{2}^{l+1})+p_{r}^{l+1}h_{k}\\&+(p_{1}^{l+1}+p_{2}^{l+1})h_{r}+1)-\log_2(p_{r}^{l}h_{k}h_{r}p_{k^{\prime}}^{l}+p_{r}^{l}h_{k}+p_{k^{\prime}}^{l}h_{r}+1)\\&-d_{k}^{l,m}(p_1^{l+1}-p_1^{l})-t_{k}^{l,m}(p_2^{l+1}-p_2^{l})-c_{k}^{l,m}(p_r^{l+1}-p_r^{l}),\\&~~~~~~~~~~~~~~~~~~~~~~~~~~~~~~~~~~(m,k)=(1,2),(m,k)=(2,1),
\label{eq:obj_fun_p_DC_p2}
\end{split}
\end{equation}
\begin{equation}
\begin{split}
&R_{k,\text{lb}}^{l+1,m}=\frac{1}{2}\bigg(\log_{2}(p_{r}^{l+1}h_{k}h_{r}p_{k}^{l+1}+p_{r}^{l+1}h_{k}+2p_{k}^{l+1}h_{r}+2)\\&-\log_2(p_{r}^{l}h_{k}+2p_{k}^{l}h_{r}+2)\\&-d_{k}^{l,m}(p_1^{l+1}-p_1^{l})-t_{k}^{l,m}(p_2^{l+1}-p_2^{l})-c_{k}^{l,m}(p_r^{l+1}-p_r^{l})\bigg),\\&~~~~~~~~~~~~~~~~~~~~~~~~~~~~~~~(m,k)=(3,1),(m,k)=(3,2),
\label{eq:obj_fun_p_DC_p3}
\end{split}
\end{equation}
where by assuming $F_{k}=\ln2(p_{r}^{l}h_{k}+(p_{1}^{l}+p_{2}^{l})h_{r}+1)$, $G_{k}=\ln2(p_{r}^{l}h_{k}h_{r}p_{k^{\prime}}^{l}+p_{r}^{l}h_{k}+p_{k^{\prime}}^{l}h_{r}+1)$, and $J_{k}=\ln2(p_{r}^{l}h_{k}+2p_{k}^{l}h_{r}+2)$, we have $d_{1}^{l,1}=\frac{h_{r}}{F_{1}}$, $t_{1}^{l,1}=\frac{h_{r}}{F_{1}}$, $c_{1}^{l,1}=\frac{h_{1}}{F_{1}}$, $d_{2}^{l,2}=\frac{h_{r}}{F_{2}}$, $t_{2}^{l,2}=\frac{h_{r}}{F_{2}}$, $c_{2}^{l,2}=\frac{h_{2}}{F_{2}}$, $d_{2}^{l,1}=\frac{p_r^lh_{2}h_{r}+h_{r}}{G_{2}}$, $t_{2}^{l,1}=0$, $c_{2}^{l,1}=\frac{p_1^lh_{2}h_{r}+h_{2}}{G_{2}}$, 
$d_{1}^{l,2}=0$, $t_{1}^{l,2}=\frac{p_r^lh_{1}h_{r}+h_{r}}{G_{1}}$, $c_{1}^{l,2}=\frac{p_2^lh_{1}h_{r}+h_{1}}{G_{1}}$,
$d_{1}^{l,3}=\frac{2h_{r}}{J_{1}}$, $t_{1}^{l,3}=0$, $c_{1}^{l,3}=\frac{h_{1}}{J_{1}}$, $d_{1}^{2,3}=0$, $t_{2}^{l,3}=\frac{2h_{r}}{J_{2}}$, and $c_{1}^{2,3}=\frac{h_{2}}{J_{2}}$.

\end{proposition}
\begin{proof}
According to Lemma \ref{DC}, the rates of vehicles at (P1.3) can be shown as a DC function. 
% Hence, the achievable rate of each vehicle is a DC function which is equivalent to the summation of one concave and one convex function.
On the other hand, we know that the first-order Taylor series expansion of a convex function provides a global lower-bound. Hence, at each iteration $l+1$ of the SCA method, we derive the first-order Taylor expansion of the convex parts of rates around the solution of the previous iteration $l$. The summation of these new affine Taylor approximations of convex parts and concave parts of rates will be a concave function as in (\ref{eq:obj_fun_p_DC_p1}), (\ref{eq:obj_fun_p_DC_p2}), and (\ref{eq:obj_fun_p_DC_p3}). Also, due to maximizing concave approximations at each iteration, the objective value of (P1.3) at the SCA method is non-decreasing with $l$. Hence, the proof is completed. 
\end{proof}
Proposition \ref{sca-power} shows that for a given power allocation at the $l^{\text{th}}$ iteration, the sum-rate of vehicles at (P1.3) is lower-bounded  by the summation of rates in (\ref{eq:obj_fun_p_DC_p1}), (\ref{eq:obj_fun_p_DC_p2}), and (\ref{eq:obj_fun_p_DC_p3}). It then follows that the optimal value of (P1.3) is lower-bounded by the optimal value of the following problem 
%\begin{subequations}
\begin{alignat}{2}
\text{(P1.4):}&~~\underset{\mathbf{P_1^{l+1}},\mathbf{P_2^{l+1}}, \mathbf{P_r^{l+1}}}{\text{max}}  ~~~~      &&\sum_{n=1}^{N}\sum_{s=1}^{S}\sum_{k=1}^{K} \alpha_{s,n}R_{k,\text{lb}}^{l+1,1}[n]\nonumber\\&&&~+\beta_{s,n}R_{k,\text{lb}}^{l+1,2}[n]+\gamma_{s,n}R_{k,\text{lb}}^{l+1,3}[n]\label{eq-power:obj_fun}\\
&\text{s.t.}&     &  \sum_{s=1}^{S} \alpha_{s,n}R_{k,\text{lb}}^{l+1,1}[n]+\beta_{s,n}R_{k,\text{lb}}^{l+1,2}[n]\nonumber\\&&&~+\gamma_{s,n}R_{k,\text{lb}}^{l+1,3}[n]\geq R_{k}^{th}, \forall n, \forall k, \label{eq-power:constraint-r1}\\
 &  &      &  (\ref{eq:constraint-ps}), (\ref{eq:constraint-pr}), (\ref{eq:constraint-noma}), (\ref{eq:constraint+}). \nonumber
\end{alignat}
Problem (P1.4) is a convex problem.
Note that we proved in Proposition \ref{sca-power} that the objective function of (P1.4) is non-decreasing over iterations and is globally upper-bounded by the optimal value of (P1.3). Therefore, the proposed sub-optimal algorithm converges. One can see the summary of the proposed iterative method for solving (P1.3) in Algorithm 2.
\begin{algorithm}
\caption{Iterative power allocation for (P1.3) with fixed trajectory planning.}\label{alg2}
\begin{algorithmic}[1]
%\For{$t = 1,2,\ldots$}
\State  Initialize the power allocation for  $\mathbf{P}_{1}^{l},~\mathbf{P}_{2}^{l},~ \text{and}~\mathbf{P}_{r}^{l}$, and let $l=0$.
\State Find the binary coefficients $\alpha_{s,n},~\beta_{s,n},~ \text{and}~ \gamma_{s,n}$, for each time slot $n$ based on the UAV fixed trajectory (i.e., $\mathbf{X}$ \text{and} $\mathbf{Y}$) and locations of vehicles and BS (according to Table \ref{table:1}).
\Repeat
\State Solve the convex problem (P1.4) by the interior-point method and find the optimal solution $\mathbf{P}_{1}^{l+1},~\mathbf{P}_{2}^{l+1},~ \text{and}~\mathbf{P}_{r}^{l+1}$.
\State Update $l=l+1$.
\Until {convergence or a maximum number of iterations is reached.}
%\EndFor
\end{algorithmic}
\end{algorithm}
\subsection{Iterative Power Allocation and Trajectory Planning}
In this subsection, we apply the AO method to solve the problem (P1). To this end, we solve the trajectory planning, operation mode selection, and power allocation sub-problems alternatively, i.e., at each iteration we first optimize the trajectory and determine operation mode coefficients. Then, using this trajectory and operation mode coefficients, we optimize the power allocation sub-problem. The associated algorithm is summarized in Algorithm 3.\par 
\begin{algorithm}
\caption{Iterative power allocation, trajectory planning and mode selection for P1.}\label{alg3}
\begin{algorithmic}[1]
%\For{$t = 1,2,\ldots$}
\State  Initialize power allocation $\mathbf{P}_{1}^{l},~\mathbf{P}_{2}^{l},~ \text{and}~\mathbf{P}_{r}^{l}$ and UAV trajectory $\mathbf{X}^{l}$ and $\mathbf{Y}^{l}$, and let $l=0$.
\State Find the binary coefficients $\alpha_{s,n}^0,~\beta_{s,n}^0~ \text{and}~ \gamma_{s,n}^0$, for each time slot $n$ based on the UAV initial trajectory (i.e., $\mathbf{X}^0$ and $\mathbf{Y}^0$) and locations of vehicles and BS (according to Table \ref{table:1}).
\Repeat
\State Solve the convex problem (P1.2) with fixed power allocations $\mathbf{P}_{1}^{l},~\mathbf{P}_{2}^{l},~ \text{and}~\mathbf{P}_{r}^{l}$ by the interior-point method, and find the optimal solution for $\mathbf{X}^{l+1}$ and $\mathbf{Y}^{l+1}$.
\State Update the binary coefficients $\alpha_{s,n}^{l+1},~\beta_{s,n}^{l+1}~ \text{and}~ \gamma_{s,n}^{l+1}$, for each time slot $n$ based on the new UAV trajectory (i.e., $\mathbf{X}^{l+1}$ and $\mathbf{Y}^{l+1}$) according to Table \ref{table:1}. 
\State Solve the convex problem (P1.4) with the fixed UAV trajectory $\mathbf{X}^{l+1}$ and $\mathbf{Y}^{l+1}$ by the interior-point method, and find the optimal solution for $\mathbf{P}_{1}^{l+1},~\mathbf{P}_{2}^{l+1},~ \text{and}~\mathbf{P}_{r}^{l+1}$.
\State Update $l=l+1$.
\Until {convergence or a maximum number of iterations is reached.}
%\EndFor
\end{algorithmic}
\end{algorithm}
 According to Proposition \ref{trajectory-iter-lb} and Proposition \ref{sca-power}, the optimal value of original problem (P1) is a global upper-bound for the optimal values of sub-problems (P1.2) and (P1.4), and hence it is also an upper-bound for the optimal value of Algorithm 3. Also, since Algorithm 3 performs Algorithm 1 and Algorithm 2 alternatively, the objective value of Algorithm 3 is non-decreasing with $l$. As a result, the proposed sub-optimal method at Algorithm 3 is guaranteed to converge. 
 %Note that since each iteration of Algorithm 3 only requires
%solving convex problems, the overall complexity
%of Algorithm 3 is polynomial in the worst scenario.
\subsection{Computational Complexity Analysis}
Now, the computational complexity of proposed iterative solutions for P1 in Algorithms 1, 2 and 3 is presented. At each iteration $l$ of AO method in Algorithm 1 and 2, computational complexity is dominated by solving convex problems in (P1.2) and (P1.4), respectively. These convex problems are solved using the interior point method. According to \cite{boyd2004convex}, the interior point method
method requires $\log(\frac{n_c}{t^0\varrho})/\log\varepsilon$ number of iterations
(Newton steps) to solve a convex problem, where $n_c$ is the total
number of constraints, $t^0$ is the initial point for approximating the accuracy of the interior-point
method, $0<\varrho\ll1$ is
the stopping criterion, and $\varepsilon$ is used for updating the accuracy
of the interior point method. For (P1.2) and (P1.4), number of constraints are $n_c^t=7N+1$ and $n_c^p=6N+2$, respectively. Note that superscripts $t$ (for trajectory) and $p$ (for power) indicate problems (P1.2) and (P1.4), respectively. Hence, the computational complexity of Algorithm 1 and 2 will be $\mathcal{O}(N_{L}^t(\frac{\log(\frac{N}{t^{0,t}\varrho^{t}})}{\log\varepsilon^{t}}))$ and $\mathcal{O}(N_{L}^p(\frac{\log(\frac{N}{t^{0,p}\varrho^{p}})}{\log\varepsilon^{p}}))$, where $N_L^t$ and $N_L^p$ are the number of iterations for convergence of Algorithm 1 and 2, respectively. At each iteration $l$ of Algorithm 3, computational complexity is dominated by solving two convex problems in (P1.2) and (P1.4) and hence, it will be $\mathcal{O}(N_{L}(\frac{\log(\frac{N}{t^{0,t}\varrho^{t}})}{\log\varepsilon^{t}}+\frac{\log(\frac{N}{t^{0,p}\varrho^{p}})}{\log\varepsilon^{p}}))$, where $N_L$ is the number of iterations for convergence of Algorithm 3.
\section{Minimum-Rate Maximization}
In this section, we aim to maximize a new objective function in which the throughput fairness  between vehicles and among different time slots is considered. Indeed, in delay-sensitive applications, we have to consider fairness and hence maximize the min-rate of vehicles. Note that based on Proposition \ref{proposition-min}, OMA has better min-rate performance in comparison to NOMA at high SNR regime, and hence we select OMA mode ($m=3$) for multiple access scheme. Assuming the same parameters with (P1), the optimization problem can be formulated as
\begin{alignat}{2}
\text{(P2):}~~&\underset{\mathbf{P}_{1},\mathbf{P}_{2},\mathbf{P}_{r},\mathbf{X},\mathbf{Y}}{\text{max}}       & &~~~~\min_{n\in \mathcal{N},k\in\mathcal{K}}~~ R_{k}^{3}[n]\label{eq:obj_funmin}\\
&\text{s.t.}&      &~~~~~ (\ref{eq:constraint-v}), (\ref{eq:constraint-s}), (\ref{eq:constraint-f}), (\ref{eq:constraint-range}),
(\ref{eq:constraint-ps}), (\ref{eq:constraint-pr}), (\ref{eq:constraint+}),  \nonumber
\end{alignat}
where $\mathcal{N}=\{1,2,...,N\}$, and $\mathcal{K}=\{1,2\}$. The objective function of (P2) is not concave with respect to trajectory and power variables. Hence, (P2) is a non-convex problem. In order to solve this problem, as with the solution to problem (P1), we use the AO method.
For the trajectory planning sub-problem with fixed power allocations, the objective function of (P2) is still a non-concave function. Therefore, we apply the SCA method in which we utilize the concave lower-bound approximation of vehicle rates derived in Proposition \ref{trajectory-iter-lb} for each iteration $l+1$ of the SCA method. Given that the minimum of concave functions is a concave function, the approximated objective function of (P2) is concave. For the power allocation sub-problem with fixed trajectory, 
the objective function of (P2) is not concave, and we apply the SCA method in which utilizing Proposition \ref{sca-power} we derive concave approximations for the rate of each vehicle. Due to the similarity of derivations with Section IV, we do not mention these derivations in this section. Finally, we apply the AO method as in Algorithm 3, to solve the problem (P2). 
\section{NUMERICAL RESULTS}
In this section, numerical results are provided to validate the proposed algorithms for trajectory planning and power allocation. The following default parameters are applied in the simulations except that we specify different values for them. We assume that the BS is located at the origin and the UAV starts and ends its flight at the same location $(x_s,y_s)=(x_f,y_f)=(200,300)$. Note that in this paper, the units of $x$ and $y$ in $(x,y)$ are in meters. We assume that UAV flight ranges are $(x_{\text{min}},x_{\text{max}})=(0,1000)$ and $(y_{\text{min}},y_{\text{max}})=(0,1000)$. The default velocity for the UAV is $v=30 ~\mathrm{m/s}$, and the UAV flies at fixed height $h=100~\mathrm{m}$. There are two vehicles in our system that move along a two-lane street in parallel with y-axis. Vehicle 1 and vehicle 2 start their path at $(x_1[0],y_1[0])=(700,100)$ and $(x_2[0],y_2[0])=(702,0)$, respectively. Vehicles move with the fixed velocity $v_1=v_2=15 ~\mathrm{m/s}$. 
The communication bandwidth is $B=10~\mathrm{MHz}$  with a
carrier frequency at 5 GHz, and a noise power spectrum density of $N_0=-174~\mathrm{dBm/Hz}$. The reference SNR at the distance $d_0=1~\mathrm{m}$ equals $\frac{\beta_0}{BN_0}=70~\mathrm{dB}$. We set the number of time slots as $N=600$ and assume that each time slot equals $\tau=0.1~\mathrm{s}$.  
%The default initial values of vehicles and UAV power are $p_1[n]=p_2[n]=0.5P_s$ and $p_r[n]=P_r$. Also, the default initial value for the UAV trajectory is the straight line between the UAV start and end points.
The threshold superiority rate at Table \ref{table:1} has default value $R_{\text{th}}^S=0.1~\mathrm{bps/Hz}$. Also, we assume two vehicles have the same minimum rate requirements $R_{1}^{t}=R_{2}^{t}=1~\mathrm{bps/Hz}$. According to this rate requirement, the default value for the maximum transmit power of the BS and UAV equals $\bar{P}_s=\bar{P}_r=0.5~\mathrm{W}$.\par
In order to solve P1 and P2, we should find a feasible initial values for power allocations and the trajectory of the UAV. For this end, we assume that the BS power is allocated equally between two vehicles and among different time slots, i.e., $p_1[n]=p_2[n]=0.5\bar{P}_s,~\forall n$. Also, we assume that the UAV power is equally allocated among different time slots, i.e., $p_r[n]=\bar{P}_r,~\forall n$. 
%These initial points for power allocations satisfy constraint (19), (20), (21) and (22) in P1 and P2.
In order to find an initial feasible trajectory, we consider straight line between the start and end locations of the UAV as an initial trajectory. 
%Note that start and end locations of UAV must satisfy $(x_f-x_s)^{2}+(y_f-y_s)^{2}<(VN\tau)^{2}$, which guarantees that UAV can reach from point $(x_s,y_s)$ to point $(x_f,y_f)$ with its maximum speed $V$ in $N\tau$ seconds. Hence, if UAV flies straight line between $(x_s,y_s)$ and $(x_f,y_f)$ with a speed $v\leq V$, constraints (15), (16), (17), and (18) will be satisfied for P1 and P2. 
With the proposed initial power allocations and trajectory, all of the constraints of P2 are satisfied and hence, these initial points are in feasible region of P2. For P1, we have one more constraint rather than P2, i.e., constraint (14), which guarantees the minimum target rate $R_{k}^{t}$  for vehicle $k$ at time slot $n$.  We use the allocated powers and trajectory derived from solving P2 as an initial point of P1. 
%Note that problem P2 finds the best uniform rate for users at different time slots, and hence, it is very suitable for satisfying constraint (14). Also, note that we do not need to wait until the convergence of P2, and at each iteration that minimum target rate of users at (14) is satisfied,  we can use power and trajectory of that iteration as initial point of P1.
\par
\begin{figure}[t]
\begin{subfigure}{0.5\textwidth}
\includegraphics[width=\linewidth]{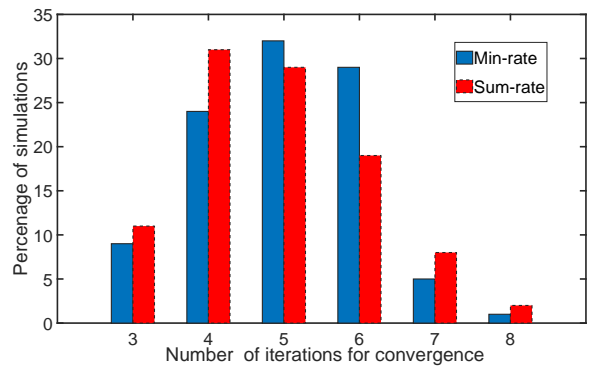} 
\caption{Histogram of the number of iterations for convergence}
\label{fig:hist}
\end{subfigure}
\begin{subfigure}{0.5\textwidth}
\includegraphics[width=\linewidth]{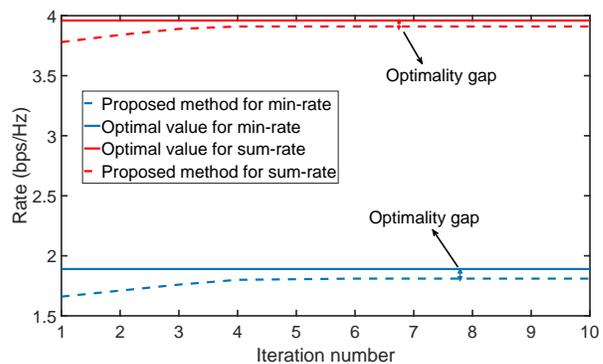}
\caption{Optimality gap}
\label{fig:optimality}
\end{subfigure}
\caption{Histogram of the number of iterations for convergence, and performance gap between proposed sub-optimal methods and optimal values for sum-rate and min-rate problems.}
\label{fig:convergence}
\end{figure}
Fig. \ref{fig:convergence} shows the histogram of the number of iterations for the convergence of the proposed algorithms and performance gap between these sub-optimal methods and optimal values. 
In Fig. \ref{fig:hist}, we see that the objective functions of (P.1) and (P.2) in (\ref{eq:obj_fun}) and (\ref{eq:obj_funmin}) converge in a few iterations which shows the efficiency of proposed algorithms. For optimization problem (P1), the convergence histogram  of Algorithm 3 has been plotted at this figure. Note that Algorithms 1 and 2 which solve problems (P1.1) and (P1.3) are special cases of Algorithm 3 when power allocation and trajectory are fixed, respectively. 
% This means that the number of iterations for convergence at Algorithm 1 and 2 is in the order of the one for Algorithm 3 in the worst case. Also, note that due to similarity of solving problem (P2) with (P1), we have not brought the algorithms for solving it. However, we can say the same explanations with optimization problem (P1) for problem (P2). 
Also, note that this histogram is based on convergence results of 100 simulations for $N=600$ time slots. In Fig. \ref{fig:optimality}, we can see the performance gap between the proposed sub-optimal algorithms and the optimal values for sum-rate and min-rate problems. Also, this figure indicates the convergence of the proposed algorithms in a few iterations. In order to attain optimal values in (P1) and (P2), we perform exhaustive search over all of the feasible values for power allocation coefficients and UAV trajectory. Because a two-dimensional search for UAV trajectory over $N$ time slots is very complex, we consider a small scale case in which we search for the optimal location for UAV placement in the XY-plane. For this end, first, we discretize the power ranges in (\ref{eq:constraint-ps}) and (\ref{eq:constraint-pr}) and the UAV flight range in (\ref{eq:constraint-range}). Then, we find the optimal power allocation and UAV deployment location that maximizes sum-rate and min-rate problems. In this figure, we assume that $(y_1,y_2)=(400,500)$. We can see in this figure that the proposed sub-optimal algorithms for (P1) and (P2) have small performance gaps with optimal values.\par
\begin{figure}[t]
\begin{subfigure}{0.5\textwidth}
\includegraphics[width=1\linewidth]{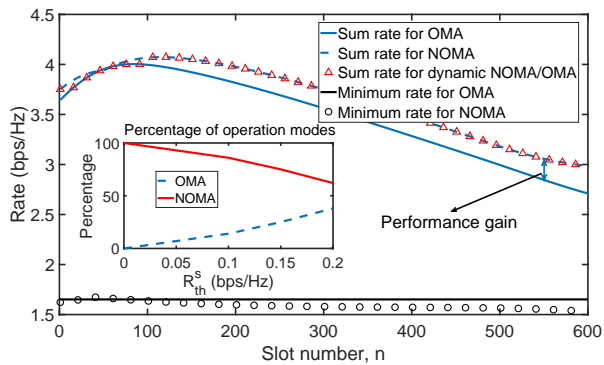} 
\caption{Sum-rate and min-rate with dynamic NOMA/OMA}
\label{fig:rates-dynamic}
\end{subfigure}
\begin{subfigure}{0.5\textwidth}
\includegraphics[width=1\linewidth]{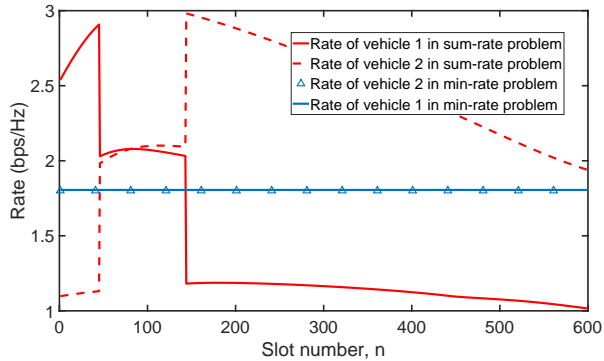}
\caption{Rates of vehicles for min-rate and sum-rate problems}
\label{rate-vehicles}
\end{subfigure}
\caption{Achievable rates with the dynamic NOMA/OMA for the sum-rate and min-rate problems versus slot number $n$.}
\label{fig:rates}
\end{figure}
Fig. \ref{fig:rates} shows the achievable rates with the OMA, NOMA, and dynamic NOMA/OMA schemes for the sum-rate and min-rate problems versus slot number $n$. Fig. \ref{fig:rates-dynamic} shows the sum-rate and min-rate of vehicles at $N=600$ time slots. In this figure, we assume that vehicle 1 and vehicle 2 start their path at $(x_1[0],y_1[0])=(700,300)$ and $(x_2[0],y_2[0])=(702,0)$, respectively. We can see that for the sum-rate problem, because two vehicles have similar distances to the UAV at initial time slots, OMA has close sum-rate performance with NOMA. Hence, our dynamic scheme selects OMA because of its lower decoding complexity. The channels of vehicles to UAV are degraded at middle and final time slots which leads to superior sum-rate performance of NOMA in comparison to OMA. Our dynamic scheme selects NOMA mode for these time slots. In this figure, we can also see the percentage of operation modes (i.e., OMA or NOMA modes) versus superiority rate ($R_{th}^{S}$). We see that by increasing $R_{th}^{S}$, the proposed dynamic scheme selects OMA scheme at more time slots. This leads to less complexity (due to lack of SIC at OMA scheme) at vehicles with compromising more sum-rate.
%We can also see that NOMA is selected as transmission scheme in 86 percent of the time slots. Indeed, our proposed dynamic scheme omits SIC decoding complexity in 14 percent of time slots with compromising negligible sum-rate. 
We can also see that min-rate of OMA  is superior to NOMA.
Fig. \ref{rate-vehicles} shows the achievable rates of each vehicle at $N=600$ time slots. For the sum-rate problem, the rates of vehicles in the initial time slots are much better than in the final time slots, and this is due to the existence of better communication links. 
%Also, for the sum-rate problem in the NOMA mode, there is a big difference between the rates of the strong vehicle and the weak vehicle in order to increase sum-rate. 
In the min-rate problem, the rates of the two vehicles are identical at all of the time slots in order to have fairness. Note that OMA mode is applied for the min-rate problem.\par
\begin{figure}[t]
\begin{subfigure}{0.5\textwidth}
\includegraphics[width=\linewidth]{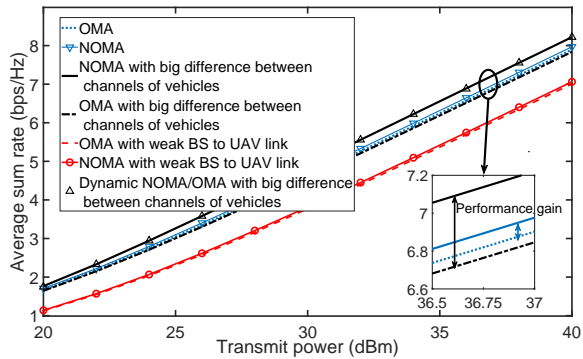} 
\caption{Average sum-rate for different channel gains of two vehicles, and weak BS-to-UAV link.}
\label{fig:sub-sum-rate-ps}
\end{subfigure}
\begin{subfigure}{0.5\textwidth}
\includegraphics[width=\linewidth]{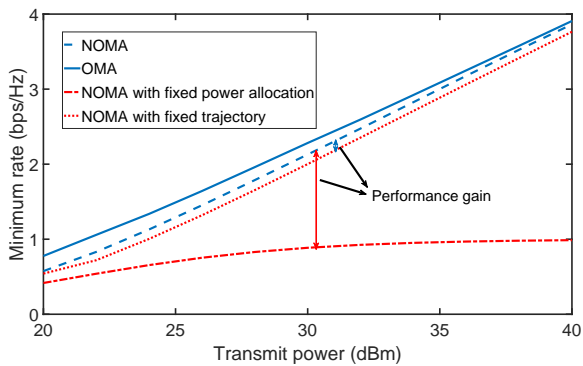}
\caption{Minimum rate for the cases where power allocation and/or trajectory are optimized.}
\label{fig:sub-min-rate-ps}
\end{subfigure}
\caption{Sum-rate and min-rate of vehicles versus transmit power with NOMA/OMA.}
\label{fig:transmit-power}
\end{figure}
Fig. \ref{fig:transmit-power} indicates the average sum-rate and min-rate of vehicles versus transmit power for the sum-rate (P1) and min-rate (P2) problems, respectively. We assumed that both the BS and UAV have the same transmit power. Fig. \ref{fig:sub-sum-rate-ps} shows the  sum-rate of vehicles averaged over $N=600$ time slots. 
%We proved in Proposition \ref{proposition-sum} that NOMA  has better sum-rate performance in comparison to OMA in the high SNR regime. 
In Fig. \ref{fig:sub-sum-rate-ps}, one can see that when vehicles are close to each other at their default paths, less gain is achieved by applying NOMA rather than OMA. This observation matches with the results of Proposition \ref{proposition-sum}. In Fig. \ref{fig:sub-sum-rate-ps}, we also assume a scenario that vehicle 1 and vehicle 2 are in two different streets and they start their path at $(x_1[0],y_1[0])=(600,0)$ and $(x_2[0],y_2[0])=(900,0)$, respectively. One can see that we obtain more gain from NOMA in comparison to OMA in this case. We can also see that the proposed dynamic NOMA/OMA scheme has the same sum-rate performance with NOMA. Note that proposed dynamic NOMA/OMA has less decoding complexity than NOMA. We also see the case where the BS-to-UAV link is worse than the UAV-to-vehicle links. We assume that the UAV is located at a fixed position $(650,500)$ in order to have weak channel with the BS. We also assume that the vehicles move along two different streets. We can see that, matching with Proposition \ref{proposition-sum}, in spite of the big difference between channels of vehicles and because of weak links between the BS and UAV, the NOMA and OMA schemes has a close sum-rate performance. Fig. \ref{fig:sub-min-rate-ps} shows the min-rate of vehicles during all $N$ time slots. We can see that OMA  has better performance than NOMA at the high SNR regime which matches with the results of Proposition \ref{proposition-min}.  We can also see the min-rate of vehicles with the NOMA scheme for the case where only the power allocation or trajectory of the UAV is optimized. We see that the min-rate of these cases is worse than the case where both trajectory and power are optimized. 
%In the fixed power allocation case, we assume that the total power is allocated equally between vehicles and among different time slots.
In the NOMA scheme, interference of the strong vehicle over the weak vehicle leads to saturation of the weak vehicle rate at the high SNRs. Hence, we see that the min-rate of the fixed power allocation case tends to a fixed value by increasing transmit power.\par 
\begin{figure}[t]
\begin{subfigure}{0.5\textwidth}
\includegraphics[ width=\linewidth]{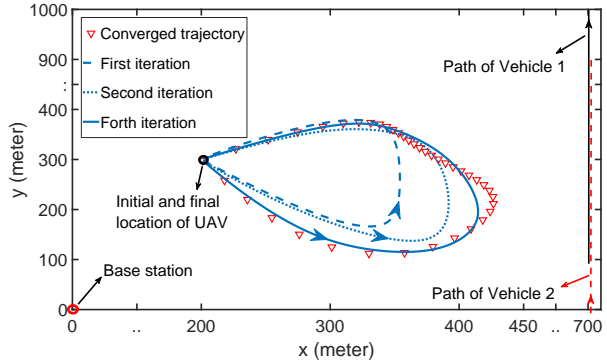} 
\caption{Evolution of UAV trajectory at min-rate problem}
\label{fig:sub-traj-min}
\end{subfigure}
\begin{subfigure}{0.5\textwidth}
\includegraphics[ width=\linewidth]{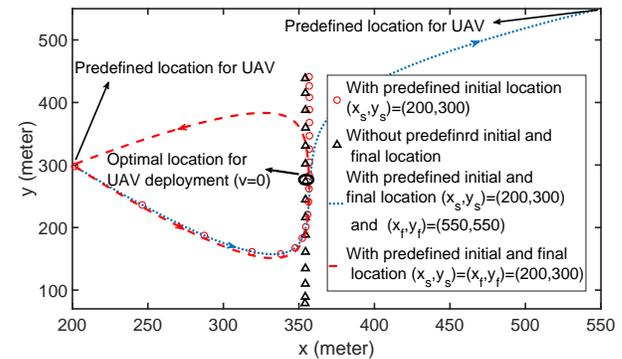}
\caption{Different scenarios at sum-rate problem}
\label{fig:sub-traj-sum}
\end{subfigure}
\caption{Trajectory of the UAV for different scenarios. (a) Evolution of UAV trajectory at the min-rate problem. (b) Trajectory for sum-rate problem when the initial and/or final points of flight are not predefined, and for two different velocities of the UAV, i.e., $v=0,30 ~\mathrm{m/s}$.}
\label{fig:trajectory}
\end{figure}
Fig. \ref{fig:trajectory} shows the optimized trajectory of the UAV and the path of vehicles at $N=600$ time slots for different scenarios.  In Fig. \ref{fig:sub-traj-min}, one can see two dimensional evolution of the UAV trajectory for the min-rate problem at different iterations of proposed algorithm. At the converged trajectory, the UAV goes forward to be close to the vehicles at middle time slots. This causes fairness in rates of vehicles between initial and middle time slots. Then, the UAV goes up to follow vehicles and provide fair rates for final time slots. Fig. \ref{fig:sub-traj-sum} indicates the trajectory of UAV for different scenarios of the sum-rate problem. For the case that initial and final locations of UAV flight are the same (i.e., $(x_s,y_s)=(x_f,y_f)=(200,300)$), the UAV flies many time slots over the line $x=356~\mathrm{m}$ in which vehicles have better sum-rate. We can see the trajectory for the case where the initial and final locations of the UAV are different points  $(x_s,y_s)=(200,300)$ and $(x_f,y_f)=(550,550)$, respectively. With these locations and velocity $v=30~\mathrm{m/s}$, the UAV is interested in spending more time slots over the line $x=356 ~\mathrm{m}$. We also solved the deployment problem for the UAV which is equivalent to the special case of the sum-rate problem where the UAV has no constraint for initial or final locations and its velocity is $v=0~\mathrm{m/s}$. In this case, the optimal point for deploying the UAV is $(x_{\text{opt}},y_{\text{opt}})=(354,264)$ which is on the line $x=354 ~\mathrm{m}$.  In Fig. \ref{fig:sub-traj-sum} we depicted the trajectory of the UAV for the case where the UAV has no predefined initial or final location. At this case, the UAV only flies over the line $x=354 ~\mathrm{m}$ following vehicles. This shows that in contrast to a fixed user case in which the UAV hovers over the best deployment location, in vehicular networks, the UAV tends to follow the vehicles to make the best performance. Finally, one can see the trajectory of the UAV when the UAV has to start its flight from a predefined location. In this case, the UAV goes to the line $x=356 ~\mathrm{m}$ and flies over this line towards vehicles.\par
\begin{figure}[t]
\begin{subfigure}{0.5\textwidth}
\includegraphics[width=\linewidth]{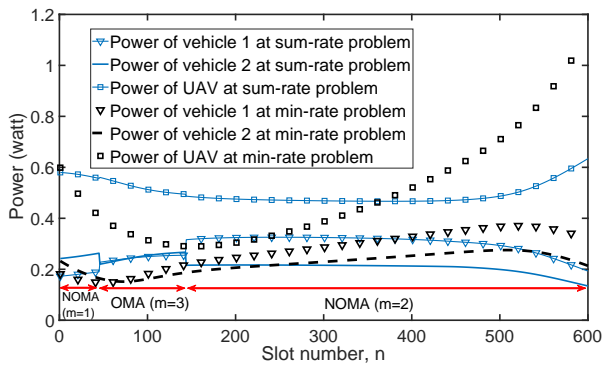} 
\caption{Min/sum-rate problems with default parameters}
\label{fig:sub-power-min-sum}
\end{subfigure}
\begin{subfigure}{0.5\textwidth}
\includegraphics[width=\linewidth]{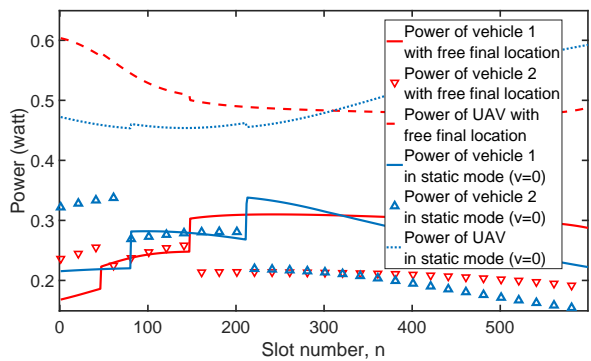}
\caption{Sum-rate problem with different scenarios.}
\label{fig:sub-power-free}
\end{subfigure}
\caption{Power allocation for the vehicles and the UAV with dynamic NOMA/OMA. (\ref{fig:sub-power-min-sum}) Min-rate and sum-rate problems with default parameters. (\ref{fig:sub-power-free}) Sum-rate problem with static relaying and the case where the UAV has no constraint for its final location.}
\label{fig:power}
\end{figure}
Fig. \ref{fig:power} indicates the allocated power with the dynamic NOMA/OMA scheme for each vehicle at the BS and UAV during the flight time. We can see these allocated powers for the sum-rate and min-rate problems with default parameters at Fig. \ref{fig:sub-power-min-sum}. For the min-rate problem, it is clear that in order to have fair rates among different time slots, the UAV and BS consume much power at final time slots in which vehicles are far from the UAV.  
%For the min-rate problem, the OMA mode is selected in which, in order to create fairness between vehicles, the vehicle with a better channel is allocated less power.
%We depicted multiple access mode at this figure by $A=\sum_{s=1}^S\alpha_{s,n}$, $B=\sum_{s=1}^S\beta_{s,n}$, and $\Gamma=\sum_{s=1}^S\gamma_{s,n}$ which are equivalent to mode 1, mode 2, and mode 3, respectively (Table \ref{table:1}).
  We can also see power allocation for the sum-rate problem in this figure. For the first time slots, due to establishing good links between the UAV and vehicles, more power is allocated to UAV to increases the sum-rate. 
  %We can see that at first time slots vehicle 1 is closer to the UAV. Hence, it performs SIC and is allocated less power. Then, vehicles have similar distances to the UAV, and hence the OMA scheme is applied. 
  %Note that against the min-rate problem, when access mode is OMA at the sum-rate problem, more power is allocated to the vehicle with better channel power gain to increase sum-rate. 
  %From time slot $n=144$, vehicle 2 has a much better link and performs SIC. 
  % Note that in spite of weak links between the UAV and vehicles, more power is allocated to the UAV at final time slots. This is because at these time slots satisfying minimum required rates for vehicles is difficult. 
  %$A=\sum_{s=1}^S\alpha_{s,n}$, $B=\sum_{s=1}^S\beta_{s,n}$, and $\Gamma=\sum_{s=1}^S\gamma_{s,n}$ which are equivalent to mode 1, mode 2, and mode 3, respectively (Table \ref{table:1}).
Fig. \ref{fig:sub-power-free} shows power allocation for the sum-rate problem with different scenarios. For the case where there is no constraint for the final location of the UAV, the UAV consumes more power at initial time slots because at these time slots the vehicles are close to the UAV. We also depicted the allocated power for the static case in which the UAV has deployed at the optimal position $(x_{\text{opt}},y_{\text{opt}})=(354,264)$. 
% Note that as one can see in Fig. \ref{fig:sub-traj-sum}, the UAV flies towards vehicles in the case that has no constraint for final location; hence, the UAV can provide the required threshold rates of vehicles with less power. However, in static mode, the UAV is far from vehicles at final time slots and we have to allocate more power for the UAV to satisfy required rates of vehicles. 
 We see in Fig. \ref{fig:sub-power-free} that the BS consumes more power at initial time slots. The reason is that at these time slots UAV-to-vehicle channels have better links and by allocating more power at the BS, better sum-rate can be achieved by the whole system. \par
\section{CONCLUSION}
 In this paper, power allocation and trajectory planning of a drone assisted NR-V2X network with dynamic NOMA/OMA has been investigated. 
We showed that in a two-user network with a dedicated AF relay, NOMA always has better or equal sum-rate performance in comparison to OMA at high SNR regime.  Due to the complexity of SIC decoding at NOMA, we proposed a dynamic NOMA/OMA scheme in which the OMA mode is selected when applying NOMA has only negligible gain.
   We have formulated two optimization problems which maximize the sum-rate and min-rate of the two vehicles. These optimization problems were not convex, and hence we applied the AO and SCA methods to find power allocation and trajectory of the UAV separately. 
    %However, these separated sub-problem were still non-convex, which led us to apply the SCA method. 
    % For the trajectory sub-problem, we found a lower-bound concave approximation for SINR expressions at each iteration. For the power allocation of the vehicles and UAV, we first approximated the rates as a DC function, and then we wrote the first order Taylor expansion for the convex part of the rates.
   % validated the convergence of these efficient  algorithms at few iterations. 
   Simulation results showed that the proposed dynamic scheme can achieve less decoding complexity in comparison to NOMA with a negligible sum-rate compromise.
% or
\appendices
\section{PROOF OF LEMMA 1}\label{lemma 1}
In order to prove that $f$ is convex, we calculate its first and second order derivatives as 
\begin{equation}
   \nabla f=
\begin{bmatrix}
\frac{\partial f}{\partial x}\\
\frac{\partial f}{\partial y}
\end{bmatrix}
=
\begin{bmatrix}
\frac{-(a+cy)}{u^2}\\
\frac{-(b+cx)}{u^2}
\end{bmatrix},
\end{equation}

\begin{equation}
\nabla^{2} f=
\begin{bmatrix}
\frac{\partial f  ^{2}}{\partial  x ^{2}} & \frac{\partial f  ^{2}}{\partial  y \partial  x} \\
\frac{\partial f^{2}}{\partial  x \partial y} & \frac{\partial f  ^{2}}{\partial  y ^{2}}
\end{bmatrix}
=\frac{1}{u^{3}}
\begin{bmatrix}
2(a+cy)^2& v\\
v&2(b+cx)^2
\end{bmatrix}, 
\end{equation}

where $u=ax+by+cxy+d$ and $v=acx+cby+c^{2}xy+2ab-cd$. In order to show that function $f$ is a convex function, we must prove that $\nabla^{2} f$ is a positive definite matrix (i.e., $\mathbf{t}\nabla^{2} f\mathbf{t}^{T}>0$ for every nonzero vector $\mathbf{t}=[t_1~~ t_2]$). Performing some calculations leads to $\mathbf{t}\nabla^{2} f\mathbf{t}^{T}=\frac{2(a+cy)^{2}t_{1}^{2}+2(b+cx)^{2}t_{2}^{2}+2(acx+cby+c^{2}xy+2ab-cd)t_{1}t_{2}}{{u^3}}$.
For writing the numerator of this equation in the square form we must have $ab=cd$. Then, we have
\begin{equation}\label{tft}
    \mathbf{t}\nabla^{2} f\mathbf{t}^{T}=
    \frac{(t_{1}(a+cy)+t_{2}(b+cx))^{2}+(a+cy)^{2}t_{1}^{2}+(b+cx)^{2}t_{2}^{2}}{{u^3}}
\end{equation}
One can see that (\ref{tft}) is always positive for $u=ax+by+cxy+d>0$. For inverse proof, assuming $ab=cd$, we can show function f as $f(x,y)=\frac{1}{((\frac{c}{b})x+1)(by+d)}$. We know that the function $g(x,y)=\frac{1}{xy}$ is a convex function for $xy>0$. On the other hand, the composition of a function with a linear function does not change convexivity. Hence, $f(x,y)=g((\frac{c}{b})x+1,by+d)$ is a convex function for $((\frac{c}{b})x+1)(by+d)>0$.
Therefore, the proof is completed.

\section{PROOF OF PROPOSITION \ref{trajectory-iter-lb} }\label{appendix-traj}
 Assuming $\psi_{r}=\frac{\sigma^2}{h_{r}}=\frac{\sigma^2(x^{2}+y^{2}+h^{2})}{\beta_{0}}$, $\psi_{i}=\frac{\sigma^2}{h_{i}}=\frac{\sigma^2((x-x_{i})^{2}+(y-y_{i})^{2}+h^{2})}{\beta_{0}}$ for $i=1,2$, and replacing them in vehicle SINR in (\ref{r1-mode1}), (\ref{r2-mode1}), (\ref{r2-mode2}), (\ref{r1-mode2}), and (\ref{rate-oma}), the SINR of vehicle $k$ equals
 \begin{equation}\label{sinr11}
    \gamma_{k}^{m}=\frac{p_{r}p_{k}}{p_{r}\psi_{r}+(p_{1}+p_{2})\psi_{k}+\psi_{r}\psi_{k}},
\end{equation} 
for $(k,m)=(1,1),~(k,m)=(2,2)$,
 \begin{equation}\label{sinr21}
   \gamma_{k}^{m}=\frac{p_{r}p_{k}}{p_{r}\psi_{r}+(p_{1}+p_{2})\psi_{k}+\psi_{r}\psi_{k}+p_{k^\prime}p_{r}},
\end{equation} 
for $(k,m)=(2,1),~(k,m)=(1,2)$, and
\begin{equation}\label{sinroma}
       \gamma_{k}^{m}=\frac{p_{r}p_{k}}{p_{r}\psi_{r}+p_{k}\psi_{k}+\psi_{r}\psi_{k}},
\end{equation}
for $(k,m)=(1,3),~(k,m)=(2,3)$,
 where $k^\prime$ shows the vehicle other than vehicle $k$. Based on Lemma \ref{lemma1}, in order to convert the SINRs (\ref{sinr11}), (\ref{sinr21}), and (\ref{sinroma}) to a convex function, we have to add fixed terms $(p_{1}+p_{2})p_{r}$, $p_{k}p_{r}$, and $p_{k}p_{r}$ to the denominators of these equations, respectively.   
 % Note that the new convex expressions for SINRs are a lower-bound approximations of the original ones. 
Note that these new SINRs are convex with respect to $\psi_{sr}$, $\psi_{r1}$, and $\psi_{r2}$, though they are not convex with respect to $\mathbf{X}$ and $\mathbf{Y}$. 
% In the SCA method, we solve the original non-convex problem by solving approximated convex problems around one initial point, iteratively. 
We know that the first-order Taylor series expansion of a convex function $f(z)$ provides a global lower-bound, i.e., $f(z)\geq f(z_0)+\nabla f(z_0)^T(z-z_0)$.
Considering this point that we have a convex approximation of SINRs with respect to $\psi_{sr}$, $\psi_{r1}$, and $\psi_{r2}$, we can approximate them at each iteration $l+1$ by their first-order Taylor series expansion around the solution of previous iteration $l$. Therefore, at the SCA method, the SINR of vehicle $k$ at the $(l+1)^{\text{th}}$ iteration can be approximated by
$\gamma_{k,\text{lb}}^{l+1,m}\approx\gamma_{k,\text{lb}}^{l,m}+d_{k,r}^{l,m}(\psi_{r}^{l+1}-\psi_{r}^{l})+d_{k,k}^{l,m}(\psi_{k}^{l+1}-\psi_{k}^{l})$.
It is clear that $\gamma_{k,\text{lb}}^{l+1,m}$ is concave with respect to $\mathbf{X}^{l+1}$ and $\mathbf{Y}^{l+1}$. Also, the composition of a concave function with a logarithm function is a concave function, and hence, $R_{k,\text{lb}}^{l+1,m}=\log_{2}(1+\gamma_{k,\text{lb}}^{l+1,m})$ is concave. Finally, note that due to applying first-order expansion for a convex function, we have $R_{k}^{l+1,m}>R_{k,\text{lb}}^{l+1,m}$, and due to the  maximization of a concave function at each iteration $l$, we have $R_{k,\text{lb}}^{l+1,m}>R_{k,\text{lb}}^{l,m}$. Considering this point that Taylor expansion of a function around an initial point is exactly equal to the value of the original function at that point (i.e., $R_{k,\text{lb}}^{l,m}=R_{k}^{l,m}$), we can conclude $R_{k}^{l+1,m}>R_{k}^{l,m}$. Hence, the rates $R_{k}^{l,m}$ are non-decreasing with $l$ and the proof is completed.

%\section{PROOF OF LEMMA \ref{con-log}}\label{proof-con-log}

%Proof of the Proposition 1

%\section*{Acknowledgment}

\ifCLASSOPTIONcaptionsoff
  \newpage
\fi



\bibliographystyle{IEEEtran}
\bibliography{main}

\begin{thebibliography}{10}
\providecommand{\url}[1]{#1}
\csname url@samestyle\endcsname
\providecommand{\newblock}{\relax}
\providecommand{\bibinfo}[2]{#2}
\providecommand{\BIBentrySTDinterwordspacing}{\spaceskip=0pt\relax}
\providecommand{\BIBentryALTinterwordstretchfactor}{4}
\providecommand{\BIBentryALTinterwordspacing}{\spaceskip=\fontdimen2\font plus
\BIBentryALTinterwordstretchfactor\fontdimen3\font minus
  \fontdimen4\font\relax}
\providecommand{\BIBforeignlanguage}[2]{{%
\expandafter\ifx\csname l@#1\endcsname\relax
\typeout{** WARNING: IEEEtran.bst: No hyphenation pattern has been}%
\typeout{** loaded for the language `#1'. Using the pattern for}%
\typeout{** the default language instead.}%
\else
\language=\csname l@#1\endcsname
\fi
#2}}
\providecommand{\BIBdecl}{\relax}
\BIBdecl

\bibitem{valavanis2015future}
K.~P. Valavanis and G.~J. Vachtsevanos, ``Future of unmanned aviation,'' in
  \emph{Handbook of Unmanned Aerial Vehicles}.\hskip 1em plus 0.5em minus
  0.4em\relax Springer, 2015, pp. 2993--3009.

\bibitem{lin2018sky}
X.~{Lin}, V.~{Yajnanarayana}, S.~D. {Muruganathan}, S.~{Gao}, H.~{Asplund},
  H.~{Maattanen}, M.~{Bergstrom}, S.~{Euler}, and Y.~E. {Wang}, ``The sky is
  not the limit: {LTE} for unmanned aerial vehicles,'' \emph{IEEE
  Communications Magazine}, vol.~56, no.~4, pp. 204--210, Apr. 2018.

\bibitem{irem}
R.~I. {Bor-Yaliniz}, A.~{El-Keyi}, and H.~{Yanikomeroglu}, ``{Efficient 3-D
  placement of an aerial base station in next generation cellular networks},''
  in \emph{2016 IEEE International Conference on Communications (ICC)}, May
  2016, pp. 1--5.

\bibitem{zeng2016throughput}
Y.~{Zeng}, R.~{Zhang}, and T.~J. {Lim}, ``Throughput maximization for
  {UAV}-enabled mobile relaying systems,'' \emph{IEEE Transactions on
  Communications}, vol.~64, no.~12, pp. 4983--4996, Dec. 2016.

\bibitem{zhan2011wireless}
P.~{Zhan}, K.~{Yu}, and A.~L. {Swindlehurst}, ``Wireless relay communications
  with unmanned aerial vehicles: Performance and optimization,'' \emph{IEEE
  Transactions on Aerospace and Electronic Systems}, vol.~47, no.~3, pp.
  2068--2085, Jul. 2011.

\bibitem{vehicular-survey}
G.~{Araniti}, C.~{Campolo}, M.~{Condoluci}, A.~{Iera}, and A.~{Molinaro},
  ``{LTE for vehicular networking: {A} survey},'' \emph{IEEE Communications
  Magazine}, vol.~51, no.~5, pp. 148--157, May 2013.

\bibitem{cheng2011infotainment}
H.~T. Cheng, H.~Shan, and W.~Zhuang, ``Infotainment and road safety service
  support in vehicular networking: From a communication perspective,''
  \emph{Mechanical Systems and Signal Processing}, vol.~25, no.~6, pp.
  2020--2038, Aug. 2011.

\bibitem{wave}
A.~{Vinel}, ``{3GPP LTE Versus IEEE 802.11p/WAVE: Which technology is able to
  support sooperative vehicular safety applications?}'' \emph{IEEE Wireless
  Communications Letters}, vol.~1, no.~2, pp. 125--128, Apr. 2012.

\bibitem{d2d-v}
L.~{Liang}, G.~Y. {Li}, and W.~{Xu}, ``Resource allocation for {D2D}-enabled
  vehicular communications,'' \emph{IEEE Transactions on Communications},
  vol.~65, no.~7, pp. 3186--3197, Jul. 2017.

\bibitem{cv2x-noma}
G.~{Liu}, Z.~{Wang}, J.~{Hu}, Z.~{Ding}, and P.~{Fan}, ``Cooperative {NOMA}
  broadcasting/multicasting for low-latency and high-reliability {5G} cellular
  {V2X} communications,'' \emph{IEEE Internet of Things Journal}, vol.~6,
  no.~5, pp. 7828--7838, Oct. 2019.

\bibitem{NR-V2X}
G.~{Naik}, B.~{Choudhury}, and J.~{Park}, ``{IEEE} 802.11bd and {5G NR V2X}:
  Evolution of radio access technologies for {V2X} communications,'' \emph{IEEE
  Access}, vol.~7, pp. 70\,169--70\,184, 2019.

\bibitem{saito2013non}
Y.~{Saito}, Y.~{Kishiyama}, A.~{Benjebbour}, T.~{Nakamura}, A.~{Li}, and
  K.~{Higuchi}, ``Non-orthogonal multiple access ({NOMA}) for cellular future
  radio access,'' in \emph{2013 IEEE 77th Vehicular Technology Conference (VTC
  Spring)}, Jun. 2013, pp. 1--5.

\bibitem{ding2016application}
Z.~{Ding}, F.~{Adachi}, and H.~V. {Poor}, ``The application of {MIMO} to
  non-orthogonal multiple access,'' \emph{IEEE Transactions on Wireless
  Communications}, vol.~15, no.~1, pp. 537--552, Jan. 2016.

\bibitem{ding2015cooperative}
Z.~{Ding}, M.~{Peng}, and H.~V. {Poor}, ``Cooperative non-orthogonal multiple
  access in {5G} systems,'' \emph{IEEE Communications Letters}, vol.~19, no.~8,
  pp. 1462--1465, Aug. 2015.

\bibitem{liu2016cooperative}
Y.~{Liu}, Z.~{Ding}, M.~{Elkashlan}, and H.~V. {Poor}, ``Cooperative
  non-orthogonal multiple access with simultaneous wireless information and
  power transfer,'' \emph{IEEE Journal on Selected Areas in Communications},
  vol.~34, no.~4, pp. 938--953, Apr. 2016.

\bibitem{abbasi2018cooperative}
O.~Abbasi and A.~Ebrahimi, ``Cooperative {NOMA} with full-duplex
  amplify-and-forward relaying,'' \emph{Transactions on Emerging
  Telecommunications Technologies}, vol.~29, no.~7, p. e3421, Jul. 2018.

\bibitem{ding2016relay}
Z.~{Ding}, H.~{Dai}, and H.~V. {Poor}, ``Relay selection for cooperative
  {NOMA},'' \emph{IEEE Wireless Communications Letters}, vol.~5, no.~4, pp.
  416--419, Aug. 2016.

\bibitem{kim2015non}
J.~{Kim} and I.~{Lee}, ``Non-orthogonal multiple access in coordinated direct
  and relay transmission,'' \emph{IEEE Communications Letters}, vol.~19,
  no.~11, pp. 2037--2040, Nov. 2015.

\bibitem{noma-inspired}
O.~{Abbasi}, A.~{Ebrahimi}, and N.~{Mokari}, ``{NOMA} inspired cooperative
  relaying system using an {AF} relay,'' \emph{IEEE Wireless Communications
  Letters}, vol.~8, no.~1, pp. 261--264, Feb. 2019.

\bibitem{choi2014energy}
D.~H. {Choi}, S.~H. {Kim}, and D.~K. {Sung}, ``Energy-efficient maneuvering and
  communication of a single {UAV}-based relay,'' \emph{IEEE Transactions on
  Aerospace and Electronic Systems}, vol.~50, no.~3, pp. 2320--2327, Jul. 2014.

\bibitem{rohde2013ad}
S.~Rohde, M.~Putzke, and C.~Wietfeld, ``{Ad hoc self-healing of OFDMA networks
  using UAV-based relays},'' \emph{Ad Hoc Networks}, vol.~11, no.~7, pp.
  1893--1906, Sep. 2013.

\bibitem{zhang2018joint}
S.~{Zhang}, H.~{Zhang}, Q.~{He}, K.~{Bian}, and L.~{Song}, ``Joint trajectory
  and power optimization for {UAV} relay networks,'' \emph{IEEE Communications
  Letters}, vol.~22, no.~1, pp. 161--164, Jan. 2018.

\bibitem{xu-joint}
X.~{Jiang}, Z.~{Wu}, Z.~{Yin}, and Z.~{Yang}, ``Joint power and trajectory
  design for {UAV}-relayed wireless systems,'' \emph{IEEE Wireless
  Communications Letters}, vol.~8, no.~3, pp. 697--700, Jun. 2019.

\bibitem{access-af-aerial}
------, ``Power and trajectory optimization for {UAV}-enabled
  amplify-and-forward relay networks,'' \emph{IEEE Access}, vol.~6, pp.
  48\,688--48\,696, 2018.

\bibitem{Nallanathan-noma-uav}
Y.~{Liu}, Z.~{Qin}, Y.~{Cai}, Y.~{Gao}, G.~Y. {Li}, and A.~{Nallanathan},
  ``{UAV} communications based on non-orthogonal multiple access,'' \emph{IEEE
  Wireless Communications}, vol.~26, no.~1, pp. 52--57, Feb. 2019.

\bibitem{sharma2017uav}
P.~K. {Sharma} and D.~I. {Kim}, ``{UAV}-enabled downlink wireless system with
  non-orthogonal multiple access,'' in \emph{2017 IEEE Globecom Workshops (GC
  Wkshps)}, Dec. 2017, pp. 1--6.

\bibitem{sohail2018non}
M.~F. {Sohail}, C.~Y. {Leow}, and S.~{Won}, ``Non-orthogonal multiple access
  for unmanned aerial vehicle assisted communication,'' \emph{IEEE Access},
  vol.~6, pp. 22\,716--22\,727, 2018.

\bibitem{poor-uav-noma}
A.~A. {Nasir}, H.~D. {Tuan}, T.~Q. {Duong}, and H.~V. {Poor}, ``{UAV}-enabled
  communication using {NOMA},'' \emph{IEEE Transactions on Communications},
  vol.~67, no.~7, pp. 5126--5138, Jul. 2019.

\bibitem{f.r.yu}
X.~{Liu}, J.~{Wang}, N.~{Zhao}, Y.~{Chen}, S.~{Zhang}, Z.~{Ding}, and F.~R.
  {Yu}, ``Placement and power allocation for {NOMA-UAV} networks,'' \emph{IEEE
  Wireless Communications Letters}, vol.~8, no.~3, pp. 965--968, Jun. 2019.

\bibitem{alo-precoding}
N.~{Zhao}, X.~{Pang}, Z.~{Li}, Y.~{Chen}, F.~{Li}, Z.~{Ding}, and M.~{Alouini},
  ``Joint trajectory and precoding optimization for {UAV}-assisted {NOMA}
  networks,'' \emph{IEEE Transactions on Communications}, vol.~67, no.~5, pp.
  3723--3735, May 2019.

\bibitem{NOMA-Enabled-V2X}
H.~{Zheng}, H.~{Li}, S.~{Hou}, and Z.~{Song}, ``Joint resource allocation with
  weighted max-min fairness for {NOMA}-enabled {V2X} communications,''
  \emph{IEEE Access}, vol.~6, pp. 65\,449--65\,462, Oct. 2018.

\bibitem{cache-v2x}
S.~{Gurugopinath}, P.~C. {Sofotasios}, Y.~{Al-Hammadi}, and S.~{Muhaidat},
  ``Cache-aided non-orthogonal multiple access for {5G}-enabled vehicular
  networks,'' \emph{IEEE Transactions on Vehicular Technology}, vol.~68, no.~9,
  pp. 8359--8371, Sep. 2019.

\bibitem{pairing}
Z.~{Ding}, P.~{Fan}, and H.~V. {Poor}, ``Impact of user pairing on {5G}
  nonorthogonal multiple-access downlink transmissions,'' \emph{IEEE
  Transactions on Vehicular Technology}, vol.~65, no.~8, pp. 6010--6023, Aug.
  2016.

\bibitem{patel2006statistical}
C.~S. {Patel}, G.~L. {Stuber}, and T.~G. {Pratt}, ``Statistical properties of
  amplify and forward relay fading channels,'' \emph{IEEE Transactions on
  Vehicular Technology}, vol.~55, no.~1, pp. 1--9, Jan. 2006.

\bibitem{benjebbour_SIC}
C.~{Yan}, A.~{Harada}, A.~{Benjebbour}, Y.~{Lan}, A.~{Li}, and H.~{Jiang},
  ``Receiver design for downlink non-orthogonal multiple access ({NOMA}),'' in
  \emph{2015 IEEE 81st Vehicular Technology Conference (VTC Spring)}, May 2015,
  pp. 1--6.

\bibitem{boyd2004convex}
S.~Boyd and L.~Vandenberghe, \emph{Convex optimization}.\hskip 1em plus 0.5em
  minus 0.4em\relax Cambridge University Press, 2004.

\end{thebibliography}



%\newpage



%\vfill


\end{document}